\documentclass[journal]{IEEEtran}
\usepackage{cite}
\usepackage{amsmath,amssymb,amsfonts,amsthm}
\usepackage{algorithmic}
\usepackage{graphicx}
\usepackage{textcomp}
\usepackage{xcolor}
\usepackage[caption=false,font=footnotesize]{subfig}

\usepackage{siunitx}

\usepackage{enumitem}

\usepackage{hyperref} 
\hypersetup{hidelinks}
\usepackage{xurl}

\definecolor{tab0}{RGB}{31, 119, 180}    
\definecolor{tab1}{RGB}{255, 127, 14}    
\definecolor{tab2}{RGB}{44, 160, 44}     
\definecolor{tab3}{RGB}{214, 39, 40}     
\definecolor{tab4}{RGB}{148, 103, 189}   
\definecolor{tab5}{RGB}{140, 86, 75}     
\definecolor{tab6}{RGB}{227, 119, 194}   
\definecolor{tab7}{RGB}{127, 127, 127}   
\definecolor{tab8}{RGB}{188, 189, 34}    
\definecolor{tab9}{RGB}{23, 190, 207}    

\usepackage{booktabs}
\usepackage{mathtools}
\usepackage{amssymb}

\newtheorem{theorem}{Theorem}
\newtheorem{lemma}{Lemma}
\newtheorem{corollary}{Corollary}
\newtheorem{proposition}{Proposition}
\newtheorem{definition}{Definition}
\newtheorem{example}{Example}

\newtheorem{assumption}{Assumption}

\newcommand{\nb}{\mathbb{N}}

\newcommand{\fb}{\mathbb{F}}

\usepackage{dsfont}

\usepackage{comment}

\DeclareMathOperator{\argmax}{argmax}
\DeclareMathOperator{\argmin}{argmin}
\DeclareMathOperator{\rank}{rank}

\DeclareMathOperator{\wt}{w_H}

\usepackage{bm}

\usepackage{acro}

\DeclareAcronym{TPC}{
  short = TPC,
  long = turbo product code,
  short-plural = s
}
\DeclareAcronym{RLC}{
  short = RLC,
  long = random linear code,
  short-plural = s
}
\DeclareAcronym{MDS}{
  short = MDS,
  long = maximum distance separable,
  short-plural = s
}
\DeclareAcronym{CRC}{
  short = CRC,
  long = cyclic redundancy check,
  short-plural = s
}
\DeclareAcronym{SCL}{
  short = SCL,
  long = successive cancellation list,
  short-plural = s
}
\DeclareAcronym{BCH}{
  short = BCH,
  long = Bose–Chaudhuri–Hocquenghem,
  short-plural = s
}
\DeclareAcronym{BCJR}{
  short = BCJR,
  long = Bahl-Cocke-Jelinek-Raviv,
  short-plural = s
}
\DeclareAcronym{BLER}{
  short = BLER,
  long = block error rate,
  short-plural = s
}
\DeclareAcronym{BI-AWGN}{
  short = BI-AWGN,
  long = binary-input additive white Gaussian noise,
  short-plural = s
}
\DeclareAcronym{SNR}{
  short = SNR,
  long = signal-to-noise ratio,
  short-plural = s
}
\DeclareAcronym{LLR}{
  short = LLR,
  long = log likelihood ratio,
  short-plural = s
}
\DeclareAcronym{SISO}{
  short = SISO,
  long = soft-input soft-output,
  short-plural = s
}

\DeclareAcronym{GRAND}{
  short = GRAND,
  long = Guessing Random Additive Noise Decoding,
  short-plural = s
}

\DeclareAcronym{SOGRAND}{
  short = SOGRAND,
  long = Soft-Output Guessing Random Additive Noise Decoding,
  short-plural = s
}

\DeclareAcronym{SO}{
  short = SO,
  long = soft-output,
  short-plural = s
}

\DeclareAcronym{LDPC}{
  short = LDPC,
  long = low-density parity-check,
  short-plural = s
}
\DeclareAcronym{LDP}{
  short = LDP,
  long = Large Deviations Principle,
  short-plural = s
}

\DeclareAcronym{VN}{
  short = VN,
  long = variable node,
  short-plural = s
}

\DeclareAcronym{CN}{
  short = CN,
  long = check node,
  short-plural = s
}

\DeclareAcronym{RS}{
  short = RS,
  long = Reed-Solomon,
  short-plural = s
}

\DeclareAcronym{ML}{
  short = ML,
  long = Maximum Likelihood,
  short-plural = s
}

\DeclareAcronym{BSC}{
  short = BSC,
  long =Binary Symmetric Channel,
  short-plural = s
}

\DeclareAcronym{DMC}{
  short = DMC,
  long =Discrete Memoryless Channel,
  short-plural = s
}

\usepackage{printlen}
\uselengthunit{mm}

\def\BibTeX{{\rm B\kern-.05em{\sc i\kern-.025em b}\kern-.08em
    T\kern-.1667em\lower.7ex\hbox{E}\kern-.125emX}}
\begin{document}

\title{Conditional Correctness in List Decoding: \\ How a Late Second Codeword Can Rescue Confidence in the First}

\author{Conrad Struss, Muriel Médard, and Ken R. Duffy%
\thanks{This work was supported by the Defense Advanced Research Projects Agency (DARPA) under Grant HR00112120008.}%
\thanks{Conrad Struss is with the Department of Mathematics, Northeastern
University, Boston, MA 02115 USA (e-mail: struss.c@northeastern.edu).}%
\thanks{Muriel Médard is with the Research Laboratory of Electronics,
Massachusetts Institute of Technology, Cambridge, MA 02139 USA
(e-mail: medard@mit.edu).}%
\thanks{Ken R. Duffy is with the Department of Mathematics and the Department
of Electrical and Computer Engineering, Northeastern University, Boston,
MA 02115 USA (e-mail: k.duffy@northeastern.edu).}}

\maketitle

\begin{abstract}
In error correction, the difference between a codeword and a received word, a candidate noise sequence, carries information on the confidence of the codeword as a proposed decoding. For additive channels, it is known that if noise sequences are ranked in decreasing order of likelihood, a lower rank correlates with higher confidence in a decoding.  
Here we substantially expand considerations by exploring the relationship between the likelihood ranks arising in list decoding and the posterior probability of decoding correctness, establishing that the ranks associated with subsequent list entries carry reliability information that is not obtainable from the first rank alone. 

Building on the recent development of Soft-Output Guessing Random Additive Noise Decoding, we first show that, under a random codebook model, soft-output expressions
for list decoding coincide with the exact finite-blocklength conditional probabilities of decoding correctness. For hard-decision maximum likelihood decoding, we characterize the asymptotic reliability of the first list entry. 
We identify a decision function whose sign determines whether the posterior correctness probability of the first entry in the list converges to one or zero and whose absolute value gives the exponential rate of convergence. Perhaps surprisingly, the asymptotic reliability of the first list entry depends only on the first, second and last ranks. 

For list size one, the result recovers the previously known likelihood-rank threshold for reliable decoding. For list size two, a sufficiently late second entry can imply that the first entry is asymptotically correct even when its rank, viewed in isolation, would indicate otherwise.

\end{abstract}


\acresetall
\section{Introduction}
Consider the canonical hard-decision forward error correction scenario in which a receiver obtains a binary word consisting of $n$ bits, $Y^n\in \fb_2^n$, that is a potentially channel-corrupted version of the transmitted codeword, $X^n \in \mathcal{C}_n \subset \fb_2^n$. Let the channel noise $N^n = Y^n - X^n \in \fb_2^n$ be independent of the transmitted word $X^n$, which is chosen uniformly at random from the codebook $\mathcal{C}_n$, and assume $\mathcal{C}_n$ contains $|\mathcal{C}_n| = 2^k$ elements. Suppose a decoder $d  \colon \fb_2^n \to \mathcal{C}_n $ returns the  codeword $ \hat{X}^n \coloneqq d(Y^n)  \in \mathcal{C}_n$.

The probability of a decoding error $P(X^n \neq \hat{X}^n)$, also known as the \ac{BLER}, has served as the standard measure of decoding performance since Shannon's pioneering work \cite{shannon1948}. It is the central object in the study of error exponents \cite{Gallager1965}, success exponents for codes operating above capacity  \cite{Arimoto1973,Dueck1979}, and in the finite blocklength regime \cite{polyanski2010finiteblocklength}. 

A less extensively studied quantity is the posterior probability that the returned codeword is the transmitted one given the received word, which is called the blockwise \ac{SO}. For codeword $ c^{n,\ast} \in \mathcal{C}_n$ and received word $y^n \in \fb_2^n$, it is defined as
\begin{align}\label{eq: definition blockwise so}
   \underbrace{ P(X^n = c^{n,\ast} \;|\; Y^n = y^n )}_{\eqqcolon P_{X^n | Y^n}(c^{n,\ast} | y^n) }  =  \frac{p_{N^n}(y^n - c^{n,\ast})}{ \displaystyle\sum_{c^n \in \mathcal{C}_n} p_{N^n}(y^n - c^n)},
\end{align}
where $p_{N^n}$ denotes the probability mass function of the binary noise $N^n$. 
In seminal work, Forney \cite{forney1968_exponential} used eq. \eqref{eq: definition blockwise so} to describe optimal decision regions for erasure and list decoding. 

The relationship between \ac{BLER} and blockwise \ac{SO} is given by
\begin{align}
    P( X^n \neq \hat{X}^n) &=1 - \sum_{y^n \in \fb_2^n} P_{X^n | Y^n}( d(y^n) |  y^n) P_{Y^n}( y^n) \nonumber\\
    &=  1 - \mathbb{E}  \left[  P_{X^n | Y^n}(\hat{X}^n | Y^n) \right] \label{eq: BLER = 1- E[SO]},
\end{align}
which makes it clear that any decoder that, for a given $y^n$, returns the codeword $c^{n,\ast}$ that maximizes eq. \eqref{eq: definition blockwise so} achieves minimal \ac{BLER}. Under the assumption of equiprobable codewords, finding the codeword that maximizes eq. \eqref{eq: definition blockwise so} is equivalent to finding 
$\argmax_{x^n \in \mathcal{C}_n} P( Y^n = y^n \;|\; X^n =  x^n )$, showing that \ac{ML} decoding minimizes the \ac{BLER} \cite[Sec. 5.2]{GallagerBook1968}.

As presented in eq. \eqref{eq: definition blockwise so}, the evaluation of blockwise \ac{SO} appears to require summation of a probabilistic calculation for all elements of the codebook, which is impractical even for moderate $k$.
Therefore, Forney \cite{forney1968_exponential} proposed the following approximation. Consider list decoding, where the decoder returns a list $\mathcal{L}= \{ \hat{c}^{n,1}, \dots, \hat{c}^{n,L}\}$ of $L$ candidate codewords. If noise sequences associated with codewords not in the list are assumed to have negligible probability, 
\begin{align}\label{eq: p not in list = 0}
    \sum_{c^n \in \mathcal{C}_n \setminus \mathcal{L}} p_{N^n}(y^n - c^n) \approx 0,
\end{align}
then a reasonable approximation of eq. \eqref{eq: definition blockwise so} is to replace the sum over the codebook with one only over the list entries. In effect, this approximation is the conditional likelihood that a codeword is the transmitted one given the transmitted one is in the list. For the condition in eq. \eqref{eq: p not in list = 0} to be appropriate, however, the list size must increase as channel conditions disimprove.
 
Recently, \ac{SOGRAND} \cite{yuanSOGRAND25} introduced a new means of analysing blockwise \ac{SO}, which led to an improved approximation to eq. \eqref{eq: definition blockwise so}. Instead of neglecting the probability mass outside the list, \ac{SOGRAND} estimates the likelihood that the transmitted codeword is not in the list. By leveraging that development and significantly expanding on it, here we generate new results and Shannon Theory related to blockwise \ac{SO} for list decoding.
 
To be precise, let $G \colon \fb_2^n \rightarrow  \{1, \dots, 2^n \} $  be a bijective function assigning to each noise  sequence $z^n \in \fb_2^n$ a rank $G(z^n ) \in \{1, \dots, 2^n \}$.
We define a \ac{ML} list decoder as follows. Let the first codeword in the list be
    \begin{align*}
      \widehat{C}^{n,1} = \argmin_{C^n \in \mathcal{C}_n} G(Y^n - C^n)
    \end{align*}
and then recursively
\begin{align*}
  \widehat{C}^{n,j} = \argmin_{C^n \in \mathcal{C}_n \setminus \{  \widehat{C}^{n,1}, \dots , \widehat{C}^{n,j-1} \}} G(Y^n - C^n)
\end{align*}
for $j  \in \{2, \dots, L\}$.

This class of decoders can be realized by \ac{GRAND}~\cite{duffy_capacity-achieving_2019}, 
which works by subtracting noise sequences  from $Y^n$ in increasing order of their $G$-rank, returning the first resulting sequence that belongs to the codebook as the decoded codeword. Continuing the search provides the list. 
If $G$ ranks noise sequences in decreasing order of likelihood, breaking ties arbitrarily, then eq. \eqref{eq: definition blockwise so} shows that this decoder corresponds to \ac{ML} decoding.

Define the ranks of the candidate noise sequences corresponding to the list entries as
\begin{align*}
    Q_j
    \coloneqq
    G(Y^n-\widehat C^{n,j}),
    \text{ for } j=1,\ldots,L,
\end{align*}
and let $z^{n,1} \coloneqq  G^{-1}(1), \dots, z^{n,2^n} \coloneqq  G^{-1}(2^n) \in \fb_2^n$ be the noise sequences in increasing order of rank. From the viewpoint of \ac{GRAND}, the likelihood ranks $Q_j$ correspond to the query indices at which codeword $j$ is found. 

Suppose we observe candidate noise sequences with ranks $Q_1 = q_1, \dots, Q_L = q_L$, then the noise sequences are $z^{n,q_j} = y^n - \widehat c^{n,j}$, for $j = 1, \dots, L$, and 
 the \ac{SOGRAND} approximation of eq. \eqref{eq: definition blockwise so}, cf. \cite[Corollary 1]{yuanSOGRAND25}, is given by 
\begin{align}\label{eq: SOGRAND formula introduction}
    \frac{ p_{N^n}(z^{n,q_j})}{\sum_{i=1}^{L}p_{N^n}(z^{n,q_i})+ \frac{2^k - L}{2^n - q_L}\left(1-\sum_{i=1}^{q_L}p_{N^n}(z^{n,i})\right)}.
\end{align}
Comparing eq. \eqref{eq: definition blockwise so} and eq. \eqref{eq: SOGRAND formula introduction}, we see that the terms in the denominators  agree for $c^n \in \mathcal{L}$. 
The power of the \ac{SOGRAND} formula lies in providing an estimate for $ \sum_{c^n \in \mathcal{C}_n \setminus  \mathcal{L}} p_{N^n}(y^n - c^n)$, the probability mass of the candidate noise sequences corresponding to codewords not in the list. In particular, one knows that by construction of the decoder, the noise sequences  $z^{n,i}$, $i \notin \{q_1, \dots, q_L\}$, $i \leq q_L$, cannot correspond to codewords. Therefore, the remaining probability mass $1-\sum_{i=1}^{q_L}p_{N^n}(z^{n,i})$ is distributed among the $2^n - q_L$ remaining ranks. If one approximates by a uniform distribution, then the candidate noise sequences corresponding to the codewords not in the list are each assigned probability $ p_{N^n}(y^n - c^n) \approx \frac{1}{2^n - q_L}\left(1-\sum_{i=1}^{q_L}p_{N^n}(z^{n,i})\right)$. As $ | \mathcal{C}_n \setminus \mathcal{L}| = 2^k - L$, this yields eq. \eqref{eq: SOGRAND formula introduction}. An alternative derivation using a geometric approximation, which assumes $2^{n-k}$ is small, can be found in \cite{yuanSOGRAND25}.

In this work, we first show in Proposition \ref{prop: exact SO proposition} that under a random codebook model, eq. \eqref{eq: SOGRAND formula introduction} is, in fact, not an approximation at all, but exactly coincides with
\begin{align}\label{eq: SO definition introduction}
    P(X^n =  \widehat{C}^{n,j} \;|\; Q_1 = q_1, \dots, Q_L = q_L),
\end{align}
the conditional probability of list entry $\widehat{C}^{n,j}$, $j = 1, \dots, L$, being the transmitted codeword, conditioned on the likelihood ranks of the candidate noise sequences corresponding to the list entries.  Being a conditional probability, eq. \eqref{eq: SOGRAND formula introduction} is thus the minimum mean square error estimate of decoding correctness among all functions of $q_1, \dots, q_L$ \cite[§9.4]{Williams_1991}.
As $G$ is allowed to be any bijection, this is true even for approximations of \ac{ML} decoding, in both hard and soft decision settings.

As eq. \eqref{eq: SOGRAND formula introduction} depends only on the channel $p_{N^n}$, the ranking function $G$, and the likelihood ranks $q_i$, it is amenable to study in the large-blocklength limit by leveraging understanding of guesswork \cite{Massey_guesswork1994,Arikan1996,Malone2004,Pfister2004,Hanawal_guesswork2011,christiansen2013_guesswork}. By the equality of eqs. \eqref{eq: SOGRAND formula introduction} and \eqref{eq: SO definition introduction}, the resulting statements concern the exact asymptotic conditional reliability, instead of approximations thereof.

For an \ac{ML} decoder it follows from eq. \eqref{eq: definition blockwise so} that a low rank indicates high reliability. For a list-size-one \ac{ML} decoder, this intuition was first made rigorous in \cite[Theorem~3]{duffy_capacity-achieving_2019}, where an asymptotic threshold was found, such that if $1/n \log Q_1$, the normalized logarithmic rank of the noise sequence corresponding to the found codeword, lies below it, the likelihood that the decoder returns the transmitted codeword converges to one. The result holds even for code rates above capacity.

A natural line of inquiry is whether, for $L>1$,  the \ac{ML}-decoded list provides additional reliability information in the form of the later ranks
$Q_2,\ldots,Q_L$ that the first rank alone cannot give. 
The central question studied in this work is how the relative positions of these ranks influence both the finite-blocklength and asymptotic reliability.

The query-centric viewpoint given by \ac{GRAND} provides a simple intuition for why the ranks of the later list entries influence the reliability of the earlier ones.
Suppose that we are interested in the reliability of the first list entry and that $Q_1$ is large, which, when considered alone, would suggest low reliability based on eq. \eqref{eq: definition blockwise so}. If the first found codeword is not the transmitted one, then the transmitted codeword must be found at query $Q_2$ or later, and no noise sequence with rank between $Q_1$ and $Q_2$ can correspond to the true channel noise.  As \ac{GRAND} queries noise sequences in decreasing order of probability, increasing $Q_2$, while holding $Q_1$ fixed, forces the true channel noise to be further into the low-probability region of noise sequences with ranks at least $Q_2$. Consequently, the probability of the first codeword being incorrect decreases as $Q_2$ increases.
Thus, a late second list entry can provide evidence that the first list entry is correct, even when $Q_1$ alone would suggest otherwise.

To make this intuition precise, we consider exponentially growing likelihood ranks, $q_j = 2^{\lfloor ny_j \rfloor}$, for $y_j \in [0,1)$ and $j = 1, \dots, L$. In Section \ref{sec: Shannon Theory}, we use large-deviations results for the noise guesswork $G(N^n)$ \cite{christiansen2013_guesswork} to identify an asymptotic decision function $\Delta$ that, given the law of the noise $N^n$ and code rate $R$, depends only on the normalized logarithms of the likelihood ranks $y_1, \dots, y_L$, and whose sign determines whether the posterior probability of the first codeword in the list being the transmitted codeword converges to $0$ or $1$. In a perhaps surprising development, Proposition \ref{prop: main proposition} establishes that the decision function depends only on $y_1,y_2$ and $y_L$, the first, second and last likelihood ranks.

For $L=1$, we recover the known threshold $y^*$ \cite[Theorem 3]{duffy_capacity-achieving_2019}, below which concentration on a correct decoding is guaranteed, and above which concentration on an incorrect decoding is guaranteed.
For $L=2$, we show that the second rank provides reliability information unavailable from the first rank alone. We identify a second threshold $\bar{y}>y^*$ such that, when $y^* <y_1 <\bar{y}$, a sufficiently high second rank implies concentration on a correct decoding, thereby showing that list decoding with list size $L\geq 2$ carries strictly more reliability information than single-codeword decoding. 

For the \ac{BSC}, we give explicit formulas for the thresholds $y^*$ and $\bar{y}$.
We show that such a rescue requires the second entry to have a rank exceeding the typical rank of the noise, $2^{nH}$, where $H$ is the Shannon entropy rate of the noise. Above $\bar{y}$, the decoding is asymptotically unreliable regardless of the second rank.
An illustration of the phenomenon for the \ac{BSC}, cf. Example \ref{ex: BSC finite n}, is given in Fig. \ref{fig:phase_diagram}.

\begin{figure}[t]
    \centering
    \includegraphics[width=1\linewidth]{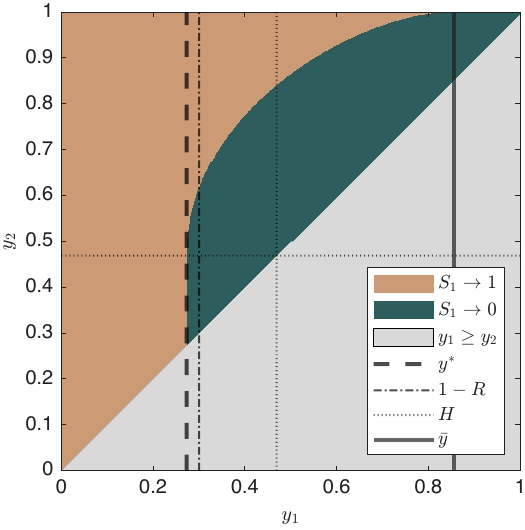}
    \caption{Visualization of Corollary \ref{cor: L=2 corollary}. Asymptotic confidence regions as code length $n \to \infty$ for the \ac{BSC} are shown in the above-capacity regime with $p = 0.1$ and $R = 0.7 > C = 1 - H = 0.531$. The orange and teal
regions indicate concentration on correct and incorrect decodings, respectively, i.e., $S_1 \coloneqq   P(X^n =  \widehat{C}^{n,1} \;|\; Q_1 = 2^{\lfloor n y_1\rfloor}, Q_2 = 2^{\lfloor n y_2\rfloor})\to 1$ and $S_1 \to 0$. The light gray region represents inadmissible ranks with $ y_1 \geq y_2$. 
The $y^* \approx 0.272$ threshold, which determines asymptotic correctness in the $L=1$ case, as well as the new threshold $\bar{y} \approx 0.854$, $H$ and $1-R$, are marked. Although $y^* <y_1$ implies asymptotic incorrectness in the $L=1$ case, in the $L=2$ case with $y^* <y_1 < \bar{y}$, asymptotic correctness can be regained if the second codeword has a sufficiently high rank. But if the first codeword has a sufficiently high rank, $\bar{y} <y_1$, asymptotic correctness cannot be regained.}
    \label{fig:phase_diagram}
\end{figure}

The rest of this paper is structured as follows.
Section \ref{sec: Background} collects background material on list decoding, guesswork, and \ac{GRAND}.
In Section \ref{sec: Soft Output exact}, we formally introduce our notation and modeling assumptions.  Proposition \ref{prop: exact SO proposition} then establishes that, under a random codebook model, and for any bijection $G$, the \ac{SOGRAND} formula eq. \eqref{eq: SOGRAND formula introduction} equals eq. \eqref{eq: SO definition introduction}, the conditional probability of correctness.
In Section \ref{sec: Shannon Theory}, we return to \ac{ML} decoding and carry out the large-blocklength analysis. Proposition \ref{prop: main proposition} establishes the decision function $\Delta$, and Corollaries \ref{cor: L=1 threshold y*} and \ref{cor: L=2 corollary} identify the thresholds $y^*$ and $\bar{y}$. We thereby show that list decoding contains strictly more reliability information than single-codeword decoding.
In Section \ref{sec: Binary Linear Codes} we exhibit numerically that for binary linear codes, the \ac{SOGRAND} formula eq. \eqref{eq: SOGRAND formula introduction} remains an accurate predictor of the exact conditional reliability eq. \eqref{eq: SO definition introduction}, and the phenomenon described in Section \ref{sec: Shannon Theory}, in which a high second noise rank increases the reliability of the first list entry, remains observable.
We conclude in Section \ref{sec: Conclusion}.

Throughout, we consider the field $\fb_2$ for ease of exposition, but note that the results generalize to higher-order finite fields.

\section{Background}\label{sec: Background}
Since its introduction by  Elias and Wozencraft \cite{elias1957_list,Wozencraft1958ListDecoding}, list decoding has been a problem of continued study. While we consider list decoding with a fixed list size $L$, in combinatorial list decoding, the term typically refers to the task of finding all codewords within a prescribed Hamming distance from the received word \cite{Guruswami2007}. Building on the list decoding algorithm for \ac{RS} codes by Guruswami and Sudan \cite{Guruswami1999}, which is confined to the hard-decision setting, Kötter and Vardy \cite{Koetter2003} introduced a soft-decision list decoder for \ac{RS} codes. List decoding is also central to the \ac{SCL} decoding of polar codes \cite{tal2015_list,niu2012crc}, where in the \ac{CRC}-aided variant, an outer \ac{CRC} code is used to select the transmitted codeword from the list obtained by the \ac{SCL} decoder. Polar codes concatenated with a \ac{CRC} have been adopted for the control channels of the 5G New Radio standard \cite{3gpp.38.212,Bioglio2021}.

\ac{GRAND} was introduced in \cite{duffy_capacity-achieving_2019} as a means of studying a code-agnostic \ac{ML} decoder. A variety of practical \ac{GRAND} algorithms have been proposed \cite{duffy2026book}, including adaptations for
soft-input decoding, e.g., \cite{Duffy19a,solomon20} and approximations of soft-input \ac{ML} decoding to allow for more efficient implementation, e.g. \cite{duffy2022_ordered,abbas2021list,abbas2025improved,yuan2026discretized}. Algorithmic advances have been 
accompanied by theoretical work ranging from performance analysis and capacity results to finite-blocklength analysis and error
exponents, e.g. \cite{liu2022orbgrand,mariona2024finite,li2024orbgrand,joudeh2024guessing,tan2025,miyamoto2026}, as well as 
implementations in circuits, e.g.
\cite{riaz2021multicodegrand,abbas2021orbgrand,condo2021fixed, chu2023efficient,blanc2024GRANDABDecoder,ji2025efficient,Riaz25ORBGRAND,Kizilates25ORBRGANDAI,Kizilates26SOGRAND}. Recently, a new channel model suitable for constructing error-correcting codes in a soft-input setting has been introduced \cite{Mariona2026}.

Central to our work is \ac{SOGRAND} \cite{yuanSOGRAND25}, which, as described in the introduction, provides blockwise \ac{SO} in the form of estimates of the a posteriori probability that the decoding is correct. The methodology behind the \ac{SOGRAND} estimate, 
inferring the likelihood of codewords not yet found from the total probability mass of the space not yet explored, has since been extended to guessing codeword decoding \cite{duffy2025soft}, successive cancellation list decoding \cite{yuan2025soft}, and bounded-distance-based list decoders \cite{Janz25soft}.

In general, applications of blockwise \ac{SO} include selecting the most reliable codeword from a  decoded list \cite{abbas2021listgrand,rapp2025reduction} and informing retransmission requests \cite{raghavan1998rova, fricke2009reliability,williamson2015variable}.  Blockwise \ac{SO} can be converted to bitwise \ac{SO} \cite{yuanSOGRAND25}, the form of soft information used to decode long, high-redundancy codes \cite{pyndiah_1998} or \ac{LDPC} codes \cite{gallagher1962_lowdensity}.

As \ac{GRAND} operates by querying sequences of noise effects in an order induced by a ranking function $G$, and the true noise is found at query $G(N^n)$, it is an instance of Massey's guesswork \cite{Massey_guesswork1994}. In \cite{Massey_guesswork1994}, Massey showed that a simple function of Shannon entropy $H$ is a lower bound for the expected number of guesses, but that ultimately, $H$ is not the correct quantity to describe the average guesswork behavior, by showing that no upper bound in terms of the Shannon entropy can exist. 

Arikan \cite{Arikan1996} established that in an asymptotic setting, it is instead the  Rényi entropy rate that characterizes the moments of guesswork. Subsequently, Arikan's work was extended from the i.i.d. case to a large class of random processes \cite{Malone2004, Pfister2004, Hanawal_guesswork2011}, which culminated in the establishment of a \ac{LDP} for guesswork \cite{christiansen2013_guesswork}. For an introduction to \acp{LDP} we refer to \cite{dembo2009large}.
Unlike the moment results, the \ac{LDP} is a distributional statement about the sequence $\{1/n \log G(N^n)\}_{n \in \mathbb{N}}$, characterizing the exponential decay rate of its rare events. This property drives our large-blocklength analysis in Section \ref{sec: Shannon Theory}. For memoryless channels, the guesswork rate function admits an explicit characterization in terms of tilted distributions \cite{Beirami2019}, which we use to derive closed-form expressions  for the \ac{BSC}.

\section{Soft Output for List Decoding}\label{sec: Soft Output exact}
The main result of this section is Proposition \ref{prop: exact SO proposition}. It codifies that, in general, the blockwise \ac{SO} formula proposed in \cite{yuanSOGRAND25} is precise for random codebooks. This equality will allow us to consider finite-$n$ results as well as large-$n$ Shannon Theory.

In the following, for a natural number $n \in \nb$, we use the notation $[n] \coloneqq \{1, \dots, n\}$ and assume that all $\log$'s are base $2$. 
We make the following modeling assumptions. Let $n,k \in \mathbb{N}$ with $k \leq n$ and
let $\fb_2^n$ be the $n$-dimensional vector space over the field with $2$ elements. We assume that the codebook $\mathcal{C}_n \coloneqq ( C^{n,1}, \dots, C^{n,2^k})$ is sampled from $\fb^n_2$ uniformly without replacement, and ordered. To be precise, if $C = ( c^{n,1}, \dots, c^{n,2^k}) \in (\fb_2^n)^{2^k}$ is any $2^k$ element ordered tuple from $\fb_2^n$, where the entries of $C$ are all distinct, then 
\begin{align*}
    P(\mathcal{C}_n = C) = \frac{1}{(2^k)! \binom{2^n}{2^k}}.
\end{align*}

We assume that the message $W$ is a uniformly distributed random variable on $[2^k]$, and that codeword $X^n \coloneqq C^{n,W}$ is sent over the channel. Let the channel act additively, meaning that the noise $N^n$ is a random variable on $\fb_2^n$ and the channel output is $Y^n = X^n + N^n.$ We also assume that $N^n,W$ and $\mathcal{C}_n$ are mutually independent.
Let $ G \colon \fb_2^n \to [2^n]$ be any bijection, and denote 
   \begin{align*}
z^{n,i} \coloneqq  G^{-1}(i)  \text{ for } i \in [2^n].
   \end{align*}  
   Then we define our list decoder as follows.
\begin{definition}\label{defi: decoder}
   Let $L \in [2^k]$ be the list size. Then define 
    \begin{align*}
      \widehat{C}^{n,1} = \argmin_{C^n \in \mathcal{C}_n} G(Y^n - C^n)
    \end{align*}
and then recursively
\begin{align*}
  \widehat{C}^{n,j} = \argmin_{C^n \in \mathcal{C}_n \setminus \{  \widehat{C}^{n,1}, \dots , \widehat{C}^{n,j-1} \}} G(Y^n - C^n)
\end{align*}
for $j  \in \{2, \dots, L\}$.
\end{definition}
Throughout this paper, we will work with the decoder defined in Definition \ref{defi: decoder}. If $G$ ranks noise sequences in decreasing order of likelihood, it is clear by definition that $\widehat{C}^{n,1}$ is the \ac{ML} decoding, which is the setting in which \ac{GRAND} was originally introduced \cite{duffy_capacity-achieving_2019}. However, we will make this assumption only in Section \ref{sec: Shannon Theory}, and for now, $G$ can be any bijection, allowing for approximations of \ac{ML} decoding, as have been proposed in practice in the soft-decision setting \cite{duffy2022_ordered,yuan2023guessing,Duffy23ORBGRANDAI,wan2024approaching}.

We define the likelihood ranks of the candidate noise sequences as 
\begin{align}\label{eq: definition Q_j}
    Q_j = G(Y^n - \widehat{C}^{n,j} ) 
\end{align}
for $j \in [L]$, or equivalently 
 \begin{align}\label{eq: second1 definition Q_j}
    Q_1 = \min_{C^n \in \mathcal{C}_n} G(Y^n - C^n ),
\end{align}
and
\begin{align}\label{eq: second2 definition Q_j}
    Q_j =\min_{C^n \in \mathcal{C}_n \setminus \{  \widehat{C}^{n,1}, \dots , \widehat{C}^{n,j-1} \}} G(Y^n - C^n).
\end{align}
Applying $G^{-1}$ to both sides of eq. \eqref{eq: definition Q_j}, we can alternatively express our decoded codewords as $ \widehat{C}^{n,j}  = Y^n - z^{n,Q_j}$.

Before stating and proving the main result of this section, we show the following auxiliary statement.
\begin{lemma}\label{lem:auxiliary lemma}
   Let  $M_n \coloneq \{ G(Y^n - C^n) \; |\; C^n \in \mathcal{C}_n \setminus \{X^n\}\}$. Then $M_n$ is, when conditioned on $G(N^n) = j$ for some $j \in [2^n]$, a uniformly random $2^k -1$ element subset of $[2^n] \setminus \{j\}$. That is, if $P(G(N^n)=j)>0$, then for any $2^k-1$ element subset $\{a_1, \dots, a_{2^k-1}\} \subset [2^n] \setminus \{j\}$ it holds that 
    \begin{align*}
        P(M_n = \{ a_1, \dots, a_{2^k -1}\} \;|\; G(N^n) = j) = \frac{1}{\dbinom{2^n-1}{2^k-1}}.
    \end{align*} 
\end{lemma}
\begin{proof}
     As $G$ is a bijection, it holds that
    \begin{align}
        &P(M_n = \{ a_1, \dots, a_{2^k -1}\} \;|\; G(N^n) = j) \nonumber\\
        &=   P(G^{-1}(M_n) = \{ 
            z^{n,a_1}, \dots, z^{n,a_{2^k -1}}\} \;|\; N^n = z^{n,j}), \label{eq:P(G^{-1}(M_n) ...}
    \end{align} 
    where $G^{-1}(M_n) \coloneqq \{ Y^n - C^n \; |\; C^n \in \mathcal{C}_n \setminus \{X^n\}\}$. Using that, conditioned on $N^n = z^{n,j}$, we have $Y^n = X^n + N^n = C^{n,W} + z^{n,j}$, and that $N^n$ is independent of $X^n$ and $\mathcal{C}_n$, and $W$ is independent of $\mathcal{C}_n$,  eq. \eqref{eq:P(G^{-1}(M_n) ...} simplifies to
    \begin{align}
        &P \big( \{ C^{n,W} - C^n \; |\; C^n \in \mathcal{C}_n \setminus \{C^{n,W}\}\} \nonumber 
         \\
         &\qquad =  \{ 
            z^{n,a_1}- z^{n,j}, \dots, z^{n,a_{2^k -1}}- z^{n,j}\}\big)\nonumber \\
            &= \sum_{w=1}^{2^k} P( \mathcal{C}_n \in B_{j,w}, W=w)  
            = \frac{1}{2^k} \sum_{w=1}^{2^k} P( \mathcal{C}_n \in B_{j,w}), \label{eq: sum_w=1 to 2k}
    \end{align}
    where 
\begin{equation*}
\begin{aligned}
B_{j,w} \coloneqq \bigl\{\, C = (c^{n,1},\dots,c^{n,2^k}) \mid {}
  &\, c^{n,w}\in \fb_2^n \text{ arbitrary}, \\
  &\, c^{n,w} - (z^{n,a_\ell}-z^{n,j}) \in C \\ 
  &\, \text{for } \ell = 1,\dots,2^k-1 \,\bigr\}.
\end{aligned}
\end{equation*}
The tuples in $B_{j,w}$ are precisely those where the $w$-th entry is arbitrary, and other entries are fixed, but their order is arbitrary. Therefore, $|B_{j,w}| = 2^n (2^k-1)!$. As $\mathcal{C}_n$ is uniformly distributed, this gives 
\begin{align*}
    P( \mathcal{C}_n \in B_{j,w}) &= \dfrac{|B_{j,w}|}{(2^k)! \dbinom{2^n}{2^k}} 
                                  = \dfrac{2^n}{2^k \dbinom{2^n}{2^k}} 
                                  = \dfrac{1}{\dbinom{2^n-1}{2^k-1}}.
\end{align*}
Plugging this into eq. \eqref{eq: sum_w=1 to 2k} concludes the proof.
\end{proof}
Now we are ready to move to the main result of this section. As $G$ is a bijection, the likelihood ranks $Q_i$ satisfy the ordering relationship $Q_1 <Q_2 < \dots < Q_{2^k}$. Therefore, in Proposition \ref{prop: exact SO proposition}, in order to not condition on events of probability zero,
we impose $q_1 < q_2 < \dots <q_L.$ Also, if $Q_L  = G(Y^n - \widehat{C}^{n,L})=q_L,$ then necessarily $G(Y^n - C^n) \in \{ q_L+1, \dots, 2^n\}$ for all $C^n \in \mathcal{C}_n \setminus \{\widehat{C}^{n,1}, \dots,\widehat{C}^{n,L} \}.$ As $|\mathcal{C}_n \setminus \{\widehat{C}^{n,1}, \dots,\widehat{C}^{n,L} \}| = 2^k - L$, we additionally impose 
\begin{align}\label{eq: well defined condition q_L}
    2^k - L \leq 2^n - q_L \iff q_L \leq 2^n - 2^k + L.
\end{align}
This is no restriction for practical applications, as $Q_L$ always satisfies eq. \eqref{eq: well defined condition q_L}, so that for $Q_L$ not satisfying eq. \eqref{eq: well defined condition q_L}, simply $P(Q_1 = q_1, \dots, Q_L = q_L) =0$.
\begin{proposition}\label{prop: exact SO proposition}
Let $L \in [2^k-1]$ and 
    $q_1, \dots, q_L \in [2^n]$ with $q_1 <q_2 < \dots < q_L$ and $q_L \leq 2^n - 2^k + L$. Then
    \begin{align}
     &P(Q_1 = q_1, \dots, Q_L = q_L)   \nonumber\\
     &= \frac{\binom{2^n-q_L}{2^k-L}}{\binom{2^n-1}{2^k-1}} \sum_{i=1}^{L} p_{N^n}  (z^{n,q_i})+\frac{\binom{2^n-1-q_L}{2^k-1-L}}{\binom{2^n-1}{2^k-1}}\sum_{i=q_L + 1}^{2^n}p_{N^n}(  z^{n,i}), \label{eq: P(Q_1=..) marginal probability}
    \end{align}
    and if $P(Q_1 = q_1, \dots, Q_L = q_L)>0$, the soft output satisfies
    \begin{align}
       &P(X^n =  \widehat{C}^{n,j} \;|\; Q_1 = q_1, \dots, Q_L = q_L) \label{eq: SO definition} \\
        &=\frac{ p_{N^n}(z^{n,q_j})}{\sum_{i=1}^{L}p_{N^n}(z^{n,q_i})+ \frac{2^k - L}{2^n - q_L}\sum_{i=q_L + 1}^{2^n}p_{N^n}(z^{n,i})} \label{eq: SO fraction 1}\\
         &= \frac{ p_{N^n}(z^{n,q_j})}{\sum_{i=1}^{L}p_{N^n}(z^{n,q_i})+ \frac{2^k - L}{2^n - q_L}\left(1-\sum_{i=1}^{q_L}p_{N^n}(z^{n,i})\right)} \label{eq: SO fraction 2}
    \end{align}
    for $j \in [L]$, where $p_{N^n}$ is the probability mass function of $N^n$.
\end{proposition}
\begin{proof}
    First note that 
    \begin{align*}
     \{ X^n= \widehat{C}^{n,j}\} = \{ X^n =Y^n- z^{n,Q_j}\}  = \{N^n = z^{n,Q_j} \},
    \end{align*}
    so that eq. \eqref{eq: SO definition} equals 
    \begin{align}\label{eq: fraction SO}
        \frac{P(N^n = z^{n,q_j}, Q_1 = q_1, \dots, Q_L = q_L)}{P(Q_1 = q_1, \dots, Q_L = q_L)}.
    \end{align}
     Note  that  $ N^n = z^{n,q_j} $ if and only if $ q_j = G(N^n)$. If $P(N^n =z^{n,q_j} ) =0,$ then eq. \eqref{eq: fraction SO} as well as eq. \eqref{eq: SO fraction 1} evaluate to $0$. So assume $P(N^n =z^{n,q_j} ) >0$ from now on.  By the definition of the $Q_j$, we can write the numerator of eq. \eqref{eq: fraction SO} as 
    \begin{align*}
        &P( Q_1 = q_1, \dots, Q_L = q_L \;|\; N^n = z^{n,q_j}) P( N^n = z^{n,q_j}) \\
        =&P(M_n \in B_{q_j} \;|\; G(N^n) = q_j) P( N^n = z^{n,q_j}),
    \end{align*}
    where $M_n$ is as in Lemma \ref{lem:auxiliary lemma}, and we define for $j \in [L-1]$
    \begin{align*}
        &B_{q_j} \coloneqq \{ \{ a_1, \dots, a_{2^k -1}\} \subset [2^n] \setminus \{q_j\}\;|\; \\
        &\qquad a_{(i)} = q_i \text{ for } i = 1, \dots, j-1,  \\
        &\qquad a_{(i)} = q_{i+1} \text{ for } i = j, \dots, L-1   \},
    \end{align*}
   which are the $2^k -1$ element subsets of $ [2^n] \setminus \{q_j\}$, where the $L-1$ smallest elements equal $q_i$ or $q_{i+1}$. Here, $a_{(i)}$ is defined as the $i$-th smallest element of $\{ a_1, \dots, a_{2^k -1}\}$. For $j = L,$ define
            \begin{align*}
        &B_{q_L} \coloneqq \{ \{ a_1, \dots, a_{2^k -1}\} \subset [2^n] \setminus \{q_L\}\;|\; \\
        &\qquad a_{(i)} = q_i \text{ for } i = 1, \dots, L-1,    \\
        & \qquad a_{(L)}>q_L\}.
    \end{align*}
    In all cases, $L-1$ elements of $B_{q_j}$ are fixed, so that $2^k -1 -  (L-1) = 2^k -L$ elements can be arbitrarily chosen from $\{q_L +1, \dots, 2^{n}\}$. Values less than $q_L+1$ are excluded as  $a_{(L)}> q_L$ for $j=L$, and this follows from $a_{(L-1)} = q_L$ for $j<L$, so in all cases, the free elements must take values  greater than $q_L$. Therefore, $|B_{q_j}| = \binom{2^n-q_L}{2^k-L} $ and Lemma \ref{lem:auxiliary lemma} yields
    \begin{align}
       &P(N^n = z^{n,q_j}, Q_1 = q_1, \dots, Q_L = q_L) \nonumber \\
       &=   P(M_n \in B_{q_j} \;|\; G(N^n) = q_j) P( N^n = z^{n,q_j}) \nonumber \\
       &= \dfrac{\binom{2^n-q_L}{2^k-L}}{\binom{2^n-1}{2^k-1}}P( N^n = z^{n,q_j}) \label{eq: P(N^n)=z, Q_1=q_1,...}.
    \end{align}
    To evaluate the denominator of eq. \eqref{eq: fraction SO}, write 
    \begin{align}
        &P(Q_1 = q_1, \dots, Q_L = q_L) \nonumber\\
        &= P(G(N^n) \leq q_L, Q_1 = q_1, \dots, Q_L = q_L) \nonumber\\
        &\qquad + P(G(N^n) > q_L, Q_1 = q_1, \dots, Q_L = q_L) \label{eq: P(Q_1 = ) split}
    \end{align}
    and note that $\{ G(N^n) \leq q_L, Q_1 = q_1, \dots, Q_L = q_L \} = \{ G(N^n) \in \{q_1, \dots, q_L\}, Q_1 = q_1, \dots, Q_L = q_L \},$ for if $G(N^n) =G(Y^n - C^{n,W}) = q$  for some $q \notin \{q_1, \dots, q_L\}$, $q <q_L$, that would imply $Y^n - z^{n,q} \in \mathcal{C}_n$ and $G(Y^n -(Y^n - z^{n,q})) = q < G(Y^n- \widehat{C}^{n,L})$, but $Y^n - z^{n,q} \notin \{\widehat{C}^{n,1}, \dots, \widehat{C}^{n,L} \}, $ a contradiction to Definition \ref{defi: decoder}. 
    Using eq. \eqref{eq: P(N^n)=z, Q_1=q_1,...}, we obtain 
    \begin{align}
        &P(G(N^n) \leq q_L, Q_1 = q_1, \dots, Q_L = q_L) \nonumber\\
        &\qquad \qquad \qquad = \frac{\binom{2^n-q_L}{2^k-L}}{\binom{2^n-1}{2^k-1}} \sum_{i=1}^{L} P( N^n = z^{n,q_i}) \label{eq: Prob G(N) < q_L,..}.
    \end{align}
    Furthermore, we have 
    \begin{align}
        &P(G(N^n) > q_L, Q_1 = q_1, \dots, Q_L = q_L) \nonumber\\
        =& \sum_{i=q_L + 1}^{2^n}  P(G(N^n) =i, Q_1 = q_1, \dots, Q_L = q_L). \label{eq: sum i>q_L+1}
    \end{align}
    Note that since $i \geq q_L+1$ for all  summands in eq. \eqref{eq: sum i>q_L+1}, they cannot be evaluated by eq. \eqref{eq: P(N^n)=z, Q_1=q_1,...}. However, an analogous argument gives 
    \begin{align*}
        &P(G(N^n) =i, Q_1 = q_1, \dots, Q_L = q_L) \\
        &= P(M_n \in \tilde{B}_{i} \;|\; G(N^n) = i) P( N^n = z^{n,i}),
    \end{align*}
    where 
    \begin{align*}
        \tilde{B}_{i} 
    &\coloneqq \{ \{ a_1, \dots, a_{2^k -1}\} \subset [2^n] \setminus \{i\}\;|\; \\
        &\qquad \qquad \qquad a_{(l)} = q_l \text{ for } l = 1, \dots, L     \}
    \end{align*}
    and we assumed $P(G(N^n)=i)>0$. The case $P(G(N^n)=i)=0$ is trivial.
    The subsets belonging to $\tilde{B}_{i}$ have their $L$ smallest components fixed to $q_1, \dots, q_L$, and the other $2^k-1-L$ elements can be arbitrarily chosen from $\{q_L+1, \dots, 2^n\}\setminus \{i\}$. Lemma \ref{lem:auxiliary lemma} yields 
    \begin{align}
        &P(G(N^n) =i, Q_1 = q_1, \dots, Q_L = q_L) \nonumber\\
        &= \frac{\binom{2^n-q_L-1}{2^k-1-L}}{\binom{2^n-1}{2^k-1}}P( N^n = z^{n,i}).\label{eq: P(G(N^n)=i,...}
    \end{align}
    Plugging eq. \eqref{eq: P(G(N^n)=i,...} into eq. \eqref{eq: sum i>q_L+1}, combining eqs. \eqref{eq: sum i>q_L+1}, \eqref{eq: Prob G(N) < q_L,..} and \eqref{eq: P(Q_1 = ) split} yields eq. \eqref{eq: P(Q_1=..) marginal probability}. Furthermore, plugging eqs. \eqref{eq: P(N^n)=z, Q_1=q_1,...} and \eqref{eq: P(Q_1=..) marginal probability} into eq. \eqref{eq: fraction SO}, we obtain that eq. \eqref{eq: fraction SO}, and thus also eq. \eqref{eq: SO definition} equals
    \begin{align*}
       \frac{\frac{\binom{2^n-q_L}{2^k-L}}{\binom{2^n-1}{2^k-1}}p_{N^n}(z^{n,q_j})}{\frac{\binom{2^n-q_L}{2^k-L}}{\binom{2^n-1}{2^k-1}} \sum_{i=1}^{L} p_{N^n}  (z^{n,q_i})+\frac{\binom{2^n-1-q_L}{2^k-1-L}}{\binom{2^n-1}{2^k-1}}\sum_{i=q_L + 1}^{2^n}p_{N^n}(  z^{n,i})} .
    \end{align*}
    Simplifying the above yields eq. \eqref{eq: SO fraction 1}, and realizing that $\sum_{i=q_L + 1}^{2^n}p_{N^n}(z^{n,i}) = 1-\sum_{i=1}^{q_L}p_{N^n}(z^{n,i}) $ yields eq. \eqref{eq: SO fraction 2}.
\end{proof}

\begin{figure*}[t]
\centering
\includegraphics[width=\textwidth]{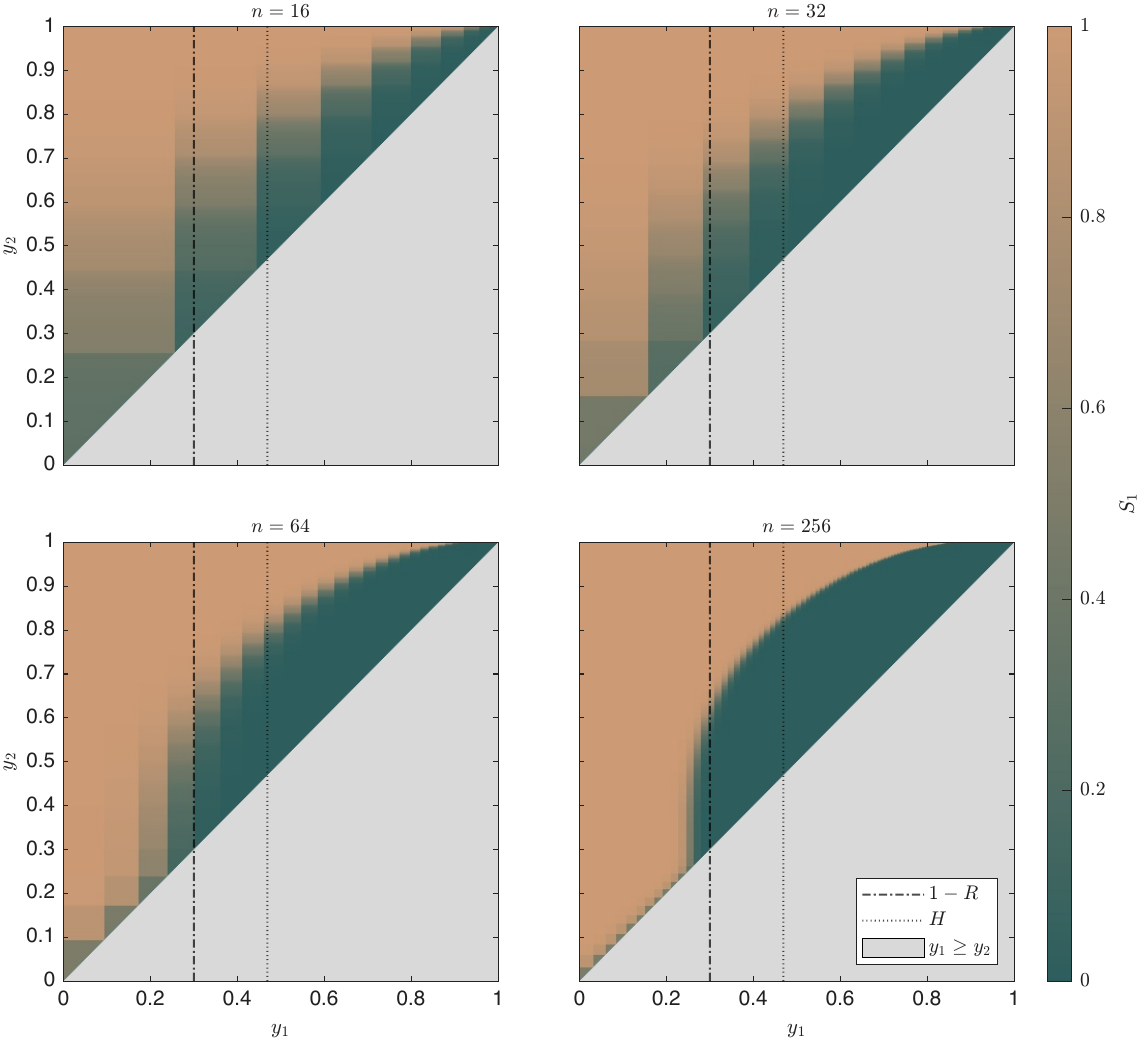}
 \caption{We show the exact finite-blocklength \ac{GRAND} \ac{SO}  $S_1 \coloneq  P(X^n =  \widehat{C}^{n,1} \;|\; Q_1 = 2^{ n y_1}, Q_2 = 2^{n y_2})$ evaluated as eq. \eqref{eq: SO fraction 2}. While Proposition \ref{prop: exact SO proposition} is stated with integer $Q_i$, we plot the continuous extension of eq. \eqref{eq: SO fraction 2} to $ny_i \in \mathbb{R}$, which agrees with conditioning on  $Q_i = 2^{\lfloor ny_i \rfloor}$ whenever $ny_i \in \mathbb{N}$. The channel is the \ac{BSC} in the above-capacity regime with $R= 0.7 >C=1-H=0.531$, $p=0.1$. Here $H = - p \log(p)- (1-p) \log (1-p)$ is the binary Shannon entropy. Tiles with $n=16,32,64$ and $n=256$ are shown.  $1-R$ and $H$ are marked. The color of each pixel indicates the value of $S_1$. The inadmissible region, $y_1 \geq y_2$, is  colored light gray. 
} 
\label{fig:finite n SO}
\end{figure*}

\begin{example}[\acs{BSC}]\label{ex: BSC finite n}
For a string $z^n \in \fb_2^n$ we define its Hamming weight as 
     $   \wt (z^n) = |\{ i \;|\; z^n_i = 1\}|$.
     Given our additive channel model $Y^n = X^n + N^n$ from above, we obtain the \ac{BSC} with crossover probability $0<p < \frac{1}{2}$,  whenever $N^n = (N_1, \dots, N_n) \in \fb_2^n$ has i.i.d. entries with
    $
        P(N_i = 1) = p
    $
    for all $i \in [n]$. Consequently, it holds  for any noise sequence $z^n \in \fb_2^n$ that 
    \begin{align*}
        P(N^n = z^n) = p^{\wt(z^n)}(1-p)^{n- \wt(z^n)}.
    \end{align*}
    Now assume also that $G$ ranks noise sequences in decreasing order of probability, breaking ties arbitrarily. Then the \ac{SO} from eq. \eqref{eq: SO fraction 2} evaluates as
     \begin{align*}
         &P(X^n =  \widehat{C}^{n,j} \;|\; Q_1 = q_1, \dots, Q_L = q_L)\\
         &= \frac{ p^{\wt (z^{n,q_j})}(1-p)^{n-\wt(z^{n,q_j})}}{\sum_{i=1}^{L}p^{\wt (z^{n,q_i})}(1-p)^{n-\wt(z^{n,q_i})}+ \frac{2^k - L}{2^n - q_L}\left(1-T_{q_L}\right)} 
    \end{align*}
    for $j \in [L]$,
    where  
    \begin{align*}
        &T_{q_L} \coloneqq \sum_{l=0}^{\wt (z^{n,q_L})-1} \binom{n}{l} p^l (1-p)^{n-l} 
        \\
        &+ \left[ q_L - \sum_{l=0}^{\wt(z^{n,q_L})-1} \binom{n}{l}\right] p^{\wt(z^{n,q_L})}(1-p)^{n- \wt(z^{n,q_L})}.
    \end{align*}
\end{example}
In Fig. \ref{fig:finite n SO} we plot the \ac{SO} from Example \ref{ex: BSC finite n}, the \ac{BSC}, for list size $L=2$, $j=1$, and increasing values of $n$. The $x$- and $y$-axes are $1/n \log q_1$ and $1/n \log q_2$, respectively. The color of each pixel corresponds to the value of the probability of the first list entry being the transmitted codeword. Two observations are immediate.
\begin{enumerate}
    \item For any $n$, if $q_1$ is fixed, then the likelihood of a correct first list entry increases as $q_2$ increases.
    \item As $n$ gets larger, the picture seems to converge to two sharply separated regions. In the first region, the likelihood of a correct first list entry is $1$, and in the second it is $0$. No likelihood values between $0$ and $1$ seem possible in the large $n$ limit, except possibly on a set of measure $0$.
\end{enumerate}
The veracity of the first observation can be verified by direct computation. In the next Section \ref{sec: Shannon Theory} we will study the limit $n\to \infty$ and see that the second observation is also true.
\section{Shannon Theory}\label{sec: Shannon Theory}
We work under the same setup as in Section \ref{sec: Soft Output exact}, but impose the following assumptions on $G$ and $N^n$ that allow for the precise quantification of rare events involving $G(N^n)$ via the techniques of large deviations theory \cite{dembo2009large}.
\begin{assumption}\label{assumption: ML optimal guessing order assumption}
  We assume that $ G \colon \fb_2^n \to [2^n]$ is a bijection that ranks the noise sequences in decreasing order of probability, breaking ties arbitrarily, i.e., for all $z^n_1, z^n_2 \in \fb_2^n$ it holds that
\begin{align*}
    P(N^n=z^n_1) > P(N^n=z^n_2) 
    \implies G(z_1^n) < G(z_2^n). 
\end{align*}
\end{assumption}
In the next Assumption \ref{assumption: sCGF regularity assumption}, we denote by 
\begin{align*}
    H_\alpha = \lim_{n\to\infty} \frac{1}{n} \frac{1}{1-\alpha}\log\left(\sum_{z^n\in\mathbb{F}_2^n} P(N^n=z^n)^\alpha\right)
\end{align*}
for $\alpha \in (0,1) \cup (1, \infty)$ the Rényi entropy rate of the noise, whenever it exists, and with $H_{\min}\coloneqq \lim_{\alpha \to \infty} H_{\alpha}$ the min-entropy rate of the noise. For $\alpha = 1$, set
\begin{align*}
  H_1 &\coloneq  \lim_{n \to \infty} \frac{1}{n} H(N^n) \\
  &= - \lim_{n \to \infty} \frac{1}{n} \sum_{z^n \in \fb_2^n} P(N^n = z^n)\log(P(N^n = z^n)),
\end{align*}
the Shannon entropy rate of the noise, which we will also denote by $H \coloneq H_1$.
\begin{assumption}\label{assumption: sCGF regularity assumption}
    Assume that \begin{align*}
\Lambda^{N}(\alpha) &\coloneqq  \lim_{n \to \infty} \frac{1}{n} \log \mathbb{E}(G(N^n)^\alpha)\\
&=
 \begin{cases} 
\alpha H_{1/(1+\alpha)} & \text{for } \alpha \in (-1, \infty) \\
-H_{\min} & \text{for } \alpha \leq -1,
\end{cases}
    \end{align*}
and that the derivative of  $\Lambda^{N}(\alpha)$ exists and is continuous on the range $\alpha \in (-1, \infty)$.
\end{assumption}
Assumption~\ref{assumption: sCGF regularity assumption} is known to be satisfied for a large class of noise processes $\{N^n\}_{n\in\mathbb{N}}$.
Arikan \cite[Proposition~5]{Arikan1996} established the case $\alpha>0$ for i.i.d. entries of $N^n$, and Malone and Sullivan \cite{Malone2004} the case $\alpha>0$ for entries forming a Markov chain. Pfister and Sullivan \cite{Pfister2004} showed that the assumption holds for $\alpha>-1$ for a large class of stationary processes, including the i.i.d. and Markov cases, and Christiansen and Duffy \cite{christiansen2013_guesswork} extended the
result from $\alpha>-1$ to $\alpha\in\mathbb{R}$.

Under Assumption \ref{assumption: sCGF regularity assumption}, the derivative of $\Lambda^N$ is continuous everywhere except possibly at $\alpha = -1$. The right-sided limit $ \gamma \coloneq \lim_{\alpha \searrow -1} \frac{d}{d\alpha}\Lambda^{N}(\alpha)$ plays a significant role in the large-deviations analysis of $1/n \log G(N^n)$ \cite{christiansen2013_guesswork}. 
 Christiansen and Duffy \cite{christiansen2013_guesswork} show that $\gamma$ is related to the exponential growth rate of the number of most likely strings. In the special case of $N^n= (N_1, \dots, N_n)$ having i.i.d. entries, \cite[Lemma 2]{christiansen2013_guesswork} yields
\begin{align}\label{eq: gamma iid noise binary}
    \gamma = \begin{cases}
        1 & \text{if } P(N_1 = 0) = P(N_1 = 1) \\
        0 & \text{if }P(N_1 = 0) \neq P(N_1 = 1).
    \end{cases}
\end{align}

We now recall the \ac{LDP} for  $\{1/n \log G(N^n)\}$, established by Christiansen and Duffy   \cite{christiansen2013_guesswork}. It serves as the foundation of our large-$n$ analysis.
\begin{theorem}\cite{christiansen2013_guesswork}\label{theo: G LDP}
    The sequence
    $\{1/n \log G(N^n)\}$ satisfies the \ac{LDP} with the convex lower semi-continuous rate function,  $I^N : \mathbb{R} \to [0, \infty]$,
\begin{align*}
I^N(x) &\coloneqq \sup_{\alpha \in \mathbb{R}} \{x \alpha - \Lambda^N(\alpha)\} \\
&=         \begin{cases}
             -x+H_{\min}, & x \in [0,\gamma] \\
             \sup_{\alpha \in \mathbb{R}} \{ x \alpha - \Lambda^{N} (\alpha)\}, & x \in (\gamma, 1] \\
             + \infty, & x \notin [0, 1].
        \end{cases}
\end{align*}
 In particular, 
 $I^N(x)$ is strictly convex on $(\gamma, 1]$ wherever it is finite, and 
 it holds for all $x \in \mathbb{R}$ that
    $
        I^N(x) = 0 \iff x =     \Lambda^N(0)' = \lim_{n \to \infty} \frac{1}{n} H(N^n) = H
    $
   and
   $
        I^N(0) = H_{\min}.
    $
\end{theorem}
Of particular interest to us is the following direct estimate on the exponential decay rate of the probability of a fixed noise sequence.
\begin{lemma}\cite[Corollary 4]{christiansen2013_guesswork}\label{lem: Guesswork paper corr}
     For $n \in \mathbb{N}$, $\epsilon>0$ and $y \in [0,1]$, let 
      \begin{align*}
          K_n(y, \epsilon) \coloneq \left\{  z^n \in \fb_2^n \; \bigg|\; \frac{1}{n} \log G(z^n) \in (y-\epsilon, y+\epsilon) \right\}.
      \end{align*}
      Then it holds for any $y \in [0,1]$ that 
      \begin{align*}
          &\lim_{\epsilon \to 0} \liminf_{n \to \infty} \frac{1}{n} \log \inf_{z^n \in K_n(y, \epsilon)}   P(N^n = z^n) \\
          &= \lim_{\epsilon \to 0} \limsup_{n \to \infty} \frac{1}{n} \log \sup_{z^n \in K_n(y, \epsilon)}   P(N^n = z^n) \\
          &= -( I^N (y) + y).
      \end{align*}
\end{lemma}

\subsection*{Results}
First we prove the following auxiliary result on the analytic properties of the noise rate function $I^N$ that we will need throughout this work. 
\begin{lemma}\label{lem: I continuous}
    The limit $H_0 \coloneqq \lim_{\alpha \to 0 } H_{\alpha}$ exists and $I^N (x)<\infty$ for all $x \in [0, H_0)$, and $I^N(x) = \infty$ for all $x > H_0$, so that $I^N$ is continuous on  $[0,H_0)$, and strictly convex on $(\gamma, H_0).$
\end{lemma}
\begin{proof}
Since the Rényi entropy is non-increasing in $\alpha$ \cite[p.14]{Polyanskiy_Wu_2025}, the same holds for the Rényi entropy rate whenever it exists, which it does by Assumption \ref{assumption: sCGF regularity assumption}. Furthermore, the Rényi entropy is bounded between $0$ and $\log 2^n = n$ \cite[p.14]{Polyanskiy_Wu_2025}, which means $0 \leq H_{\alpha} \leq 1$. Therefore, $H_0 \coloneqq \lim_{\alpha \to 0 } H_{\alpha} \in [0,1]$ exists.

To show the second part, note that by Assumption \ref{assumption: sCGF regularity assumption} we have 
\begin{align}\label{eq: I^n definition}
    I^N(x) = \sup_{\alpha \in \mathbb{R}} \{x \alpha - \Lambda^N(\alpha)\}
\end{align}
with $\Lambda^N(\alpha) = \alpha H_{1/(1+\alpha)}$ for $\alpha >-1,$ and $\Lambda^N(\alpha) = -H_{\min}$ for $\alpha \leq -1.$
If $\alpha \leq -1,$ then the term inside the supremum in eq. \eqref{eq: I^n definition} is bounded above by $-x + H_{\min}$, and  $-x + H_{\min} =  \lim_{\alpha \searrow -1} \alpha x  - \alpha H_{1/(1+\alpha)}$. So in particular,
\begin{align}\label{eq: I^n sup alpha >-1}
    I^N(x) = \sup_{\alpha > -1} \left\{\alpha \left(x  -  H_{1/(1+\alpha)} \right)\right\}.
\end{align}

Now let $x \in [0,H_0)$. By Assumption \ref{assumption: sCGF regularity assumption}, the term inside the supremum in eq. \eqref{eq: I^n sup alpha >-1} is continuous in $\alpha$, so it is bounded on compact intervals. As  $\lim_{\alpha \searrow -1} \alpha \left(x  -  H_{1/(1+\alpha)} \right) = -x + H_{\min},$ to show that $I^N(x)< \infty,$ it suffices to show that $\limsup_{\alpha \to \infty} \alpha \left(x  -  H_{1/(1+\alpha)} \right) < \infty.$ Let $\epsilon>0$ be such that $x+\epsilon < H_0$. Then, since $H_{1/(1+\alpha)} \xrightarrow[]{ \alpha \to \infty} H_0 > x + \epsilon$, it holds for all $\alpha>0$ large enough that  $\alpha \left(x  -  H_{1/(1+\alpha)} \right) < \alpha \left(x  -  (x+ \epsilon) \right)= - \alpha\epsilon \leq 0 . $ It follows that $I^N(x) < \infty.$

Next, let $x>H_0$. Then there exists an $\epsilon >0$ such that $x-\epsilon >H_0$. For analogous reasons to those above, it holds that $\alpha \left(x  -  H_{1/(1+\alpha)} \right) \geq \alpha \left(x  -  (x-\epsilon) \right)= \alpha \epsilon$ for all sufficiently large $\alpha$. But, as the right-hand side is unbounded in $\alpha$, it follows that $I^N(x) = \infty.$

The strict convexity on $(\gamma, H_0)$ then follows from Theorem \ref{theo: G LDP}. As $I^N$ is convex, lower semi-continuous and finite in $[0,H_0)$, it is continuous in $[0,H_0)$ \cite[Theorem 10.2]{Rockafellar1970}.
\end{proof}
Besides establishing the analytic properties of the rate function $I^N$, Lemma \ref{lem: I continuous} reveals the role that the Hartley entropy rate
$H_0$ plays in our analysis. Since $I^N$ is the rate function of the \ac{LDP} satisfied by $\frac{1}{n}\log G(N^n)$, the probability that
the noise has rank larger than $2^{nx}$ decays exponentially in $n$ for
$H<x<H_0$, but super-exponentially for $x>H_0$, where $I^N(x)=\infty$.

Beyond that, Lemma \ref{lem: I continuous} makes no statement about  the value of $I^N$ at the boundary point $H_0$. The value $I^N(H_0)$ plays an important role in one of our main results, Corollary \ref{cor: L=2 corollary}. We will see that when $I^N(H_0) =\infty$, a sufficiently late second list entry can rescue confidence in the first entry for every
$y_1\in(y^*,H_0)$. But if $I^N(H_0)< \infty$, then there exists a second threshold $\bar{y}>y^*$ such that a rescue becomes impossible once $y_1 \in  (\bar{y},H_0)$.
We give the following characterization of $I^N(H_0)$ as the negative of the right-sided derivative of the Rényi entropy rate $\alpha\mapsto H_\alpha$ at $\alpha =0$.
\begin{lemma}\label{lem: I(H_0)}
    It holds that 
    \begin{align}
        \lim_{y \nearrow H_0} I^N(y) &=I^N(H_0) \label{eq: limit I^n y to H_0 1}\\
        &=\lim_{\alpha \to \infty} \alpha (H_0 - H_{1/(1+\alpha)})\label{eq: limit I^n y to H_0 2}\\
        &= -  \frac{d}{d\alpha} H_{\alpha}\big|_{0^+}, \label{eq: limit I^n y to H_0 3}
    \end{align}
   where the limits exist in $[0, \infty]$. 
\end{lemma}
\begin{proof}
    By lower semi-continuity, 
    \begin{align}\label{eq: lower semi-continuous}
        I^N(H_0) \leq \liminf_{y \nearrow H_0} I^N(y).
    \end{align}
    So if $I^N(H_0) = \infty$, eq. \eqref{eq: limit I^n y to H_0 1}  follows immediately. Let $I^N(H_0) < \infty$. Due to eq. \eqref{eq: lower semi-continuous}, it suffices to show 
        \begin{align}\label{eq: upper semi-continuous}
         \limsup_{y \nearrow H_0} I^N(y)  \leq I^N(H_0).
    \end{align}
    By Lemma \ref{lem: I continuous}, and since $I^N(H_0) < \infty$, it holds that $I^N(y) < \infty$ for all $y \in [0,H_0]$. Take some $x <H_0$. Then every $y \in (x,H_0)$ can be written as $y= (1-t_y)x + t_y H_0$ for  $t_y \coloneqq (y-x)/(H_0-x)$. In particular, $0<t_y<1$ and  $t_y \to 1$ as $y \nearrow H_0$. Convexity of $I^N$ implies 
    \begin{align*}
        I^N(y) \leq (1-t_y)I^N(x)+t_y I^N(H_0).
    \end{align*}
   Taking $\limsup$ on both sides, we obtain
        \begin{align*}
        \limsup_{y \nearrow H_0}  I^N(y) &\leq  \limsup_{y \nearrow H_0}(1-t_y)I^N(x)+t_y I^N(H_0) \\
        &= I^N(H_0).
    \end{align*}
    This shows eq. \eqref{eq: limit I^n y to H_0 1}. Consider the representation of $I^N$ in eq. \eqref{eq: I^n sup alpha >-1}. Since $\alpha \mapsto H_{\alpha}$ is non-increasing, it holds that $H_0 - H_{1/(1+\alpha)} \geq0$ and thus $\alpha (H_0 - H_{1/(1+\alpha)}) \leq 0$ for all $-1<\alpha \leq 0$. Since $\alpha (H_0 - H_{1/(1+\alpha)})=0$ for $\alpha=0$, we can further simplify eq. \eqref{eq: I^n sup alpha >-1} to obtain
\begin{align}\label{eq: I^n sup alpha >0}
    I^N(H_0) = \sup_{\alpha \geq0} \left\{\alpha \left(H_0  -  H_{1/(1+\alpha)} \right)\right\}.
\end{align}
To show that eq. \eqref{eq: I^n sup alpha >0} equals eq. \eqref{eq: limit I^n y to H_0 2}, it suffices to show that $f(\alpha) \coloneqq  \alpha \left(H_0  -  H_{1/(1+\alpha)}  \right) = \alpha H_0 - \Lambda^N(\alpha) $ is non-decreasing in $[0,\infty).$ As $\Lambda^N$ is the scaled cumulant generating function of $\{ 1/n \log G(N^n)\}$, it is convex \cite[Lemma 2.3.9]{dembo2009large}. Moreover, it is differentiable by Assumption \ref{assumption: sCGF regularity assumption}. Since $\Lambda^N(\alpha)/\alpha = H_{1/(1+\alpha)} \xrightarrow{\alpha \to \infty} H_0$, it follows for all $\alpha, \beta \in [0,\infty)$, $\beta > \alpha$, that 
\begin{align*}
{\Lambda^{N}}' (\alpha) \leq \frac{\Lambda^{N}(\beta) - \Lambda^{N}(\alpha)}{\beta - \alpha} &= \frac{\frac{\Lambda^{N}(\beta)}{\beta} - \frac{\Lambda^{N}(\alpha)}{\beta}}{1-\frac{\alpha}{\beta}}\\
&\xrightarrow{\beta \to \infty} H_0.
\end{align*}
Therefore, $f(\alpha)' = H_0 - {\Lambda^{N}}' (\alpha) \geq 0$ for all $\alpha \in [0, \infty)$ and $f$ is non-decreasing. Finally, eq. \eqref{eq: limit I^n y to H_0 3} follows from eq. \eqref{eq: limit I^n y to H_0 2} by noting that 
\begin{align*}
   \alpha (H_0 - H_{1/(1+\alpha)}) = - \frac{1/(1+\alpha)}{1/\alpha}\frac{H_0 - H_{1/(1+\alpha)}}{0- 1/(1+\alpha)} 
\end{align*}
and then taking the limit $\alpha \to \infty$.
\end{proof}
Before moving on to  our main results, we continue the \ac{BSC} example. 
\begin{example}[\acs{BSC} continued]\label{ex: BSC gamma, H_0}
In the \ac{BSC} with $0<p<\frac{1}{2}$, we have 
\begin{align*}
    H_{\alpha} = \frac{1}{1-\alpha} \log \left( p^{\alpha} + (1-p)^{\alpha} \right)
\end{align*}
for $\alpha \in (0,1) \cup (1,\infty)$. Consequently, a computation gives $H_{\min} = - \log(1-p)$, $H_0 = 1$, eq. \eqref{eq: gamma iid noise binary} gives $\gamma =0$, and eq. \eqref{eq: limit I^n y to H_0 2} or eq. \eqref{eq: limit I^n y to H_0 3} give  $I^N(H_0)= - \log\left(2 \sqrt{p(1-p)}\right).$
\end{example}
Studying the \ac{SOGRAND} formula eq. \eqref{eq: SO fraction 1}, we notice that in order to study its large-blocklength limit, two non-trivial objects are of interest. The first is the exponential decay rate of $ p_{N^n}(z^{n, 2^{\lfloor ny \rfloor} }),$ the probability of a noise sequence at exponential rank $2^{\lfloor ny \rfloor} $. This rate is given by Lemma \ref{lem: Guesswork paper corr}, which is itself a corollary of the guesswork \ac{LDP} \cite[Corollary 4]{christiansen2013_guesswork}. The second quantity of interest is the tail probability $P(G(N^n) \geq 2^{ \lfloor n y \rfloor}+ 1)= \sum_{i=
  2^{ \lfloor n y \rfloor}+ 1}^{2^n}p_{N^n}(z^{n,i}).$  Here, a direct application of the \ac{LDP}, together with the continuity of $I^N$ established
in Lemma \ref{lem: I continuous}, provides the exponential decay rate, exhibiting the utility of the \ac{LDP} as a distributional statement.
\begin{lemma}\label{lem: main lemma}
  For any $y \in [0,H_0)$,  it holds that
    \begin{align}\label{eq: limit P(1/n log G equal y)}
 \lim_{n \to \infty} \frac{1}{n} \log  p_{N^n}(z^{n, 2^{\lfloor ny \rfloor} })  = -(I^N(y)+y)
    \end{align}
    and 
    \begin{align}\label{eq: limit P(1/n log G bigger y)}
  \lim_{n \to \infty} \frac{1}{n}\log \left(     \sum_{i=
  2^{ \lfloor n y \rfloor}+ 1}^{2^n}p_{N^n}(z^{n,i}) \right) =   - \inf_{x \geq y} I^N(x).
    \end{align}
\end{lemma}
\begin{proof}
     As $\frac{1}{n} \log  2^{\lfloor ny \rfloor} \to y$, we have for all $n$ large enough that $z^{n,2^{\lfloor ny \rfloor}} \in  K_n(y, \epsilon)$, where  $K_n(y, \epsilon)$ is as defined in Lemma \ref{lem: Guesswork paper corr}. So in particular
   $
       \inf_{z^n \in K_n(y, \epsilon)}   P(N^n = z^n) \leq  p_{N^n}(z^{n, 2^{ \lfloor n y \rfloor} }) \leq  \sup_{z^n \in K_n(y, \epsilon)}  P(N^n = z^n),
   $
  and  Lemma \ref{lem: Guesswork paper corr} 
  yields 
    eq. \eqref{eq: limit P(1/n log G equal y)}.
    
    To show eq. \eqref{eq: limit P(1/n log G bigger y)},  note that  $\sum_{i= 2^{ \lfloor ny \rfloor}  + 1}^{2^n}p_{N^n}(z^{n,i})  = P(1/n \log G(N^n) >  \lfloor ny \rfloor /n)$. Let $\epsilon>0.$ Then it holds for all $n$ large enough that
    \begin{align*}
      P(1/n \log G(N^n) > y) &\leq   \sum_{i= 2^{ \lfloor ny \rfloor}  + 1}^{2^n}p_{N^n}(z^{n,i})  \\
      &\leq P(1/n \log G(N^n) > y-\epsilon).
    \end{align*}
    By definition of the \ac{LDP} \cite[p. 5]{dembo2009large}, we have 
    \begin{align*}
        - \inf_{x \in (y,1]} I^N(x) &\leq \liminf_{n \to \infty} \frac{1}{n} \log P( 1/n \log G(N^n)>y) \\
        &\leq \liminf_{n \to \infty} \frac{1}{n} \log \sum_{i= 2^{ \lfloor ny \rfloor}  + 1}^{2^n}p_{N^n}(z^{n,i}) .
    \end{align*}
    Moreover, \begin{align*}
      &\limsup_{n \to \infty} \frac{1}{n} \log \sum_{i= 2^{ \lfloor ny \rfloor}  + 1}^{2^n}p_{N^n}(z^{n,i}) \\
      &\leq \limsup_{n \to \infty} \frac{1}{n} \log P(1/n \log G(N^n) > y-\epsilon) \\
        &\leq    - \inf_{x \in [y-\epsilon,1]} I^N(x).
    \end{align*}
    By Lemma \ref{lem: I continuous}, $I^N$ is continuous in $y$. So taking $\epsilon \rightarrow0 $ and noting that $ \inf_{x \in (y,1]} I^N(x) =  \inf_{x \in [y,1]} I^N(x)$ gives  eq. \eqref{eq: limit P(1/n log G bigger y)}.
\end{proof}

We now turn to the main results of this work.

\begin{proposition}\label{prop: main proposition}
Suppose $k = \lfloor nR \rfloor $ for some $R \in (0,1).$ 
Let 
\begin{align*}
    &S(2^{ny_1 }, \dots, 2^{n y_L}) \\
    &\coloneqq P(X^n =  \widehat{C}^{n,1} \;|\; Q_1 = 2^{\lfloor ny_1 \rfloor }, \dots, Q_L = 2^{\lfloor ny_L \rfloor })
\end{align*}
be the probability of the first codeword in the list being correct, given that codeword $i$ has corresponding likelihood rank $2^{\lfloor ny_i \rfloor}.$
Let $y_1, y_2, \dots, y_L \in [0,H_0)$  be such that $y_1< y_2 <\dots< y_L$.
 
For $L=1$, define $\Delta(y_1) \coloneqq   1-R + \inf_{x \geq y_1} I^N(x) - (I^N(y_1) +y_1)$, and for $L\geq 2$, define 
\begin{align*}
    &\Delta(y_1, \dots,y_L) \\
    &\coloneqq    \min \left\{ I^N(y_2) +y_2 , 1-R + \inf_{x \geq y_L} I^N(x)  \right\}\\
    &\qquad \qquad  - (I^N(y_1) +y_1 ).
\end{align*}
Then
\begin{align}\label{eq: lim 1/n log (S) ->1}
 \Delta(y_1, \dots,y_L) >0 \implies S(2^{ny_1 }, \dots, 2^{n y_L}) \xrightarrow{n \to \infty} 1
\end{align}
and
\begin{align}\label{eq: lim 1/n log (1-S) ->0}
  \Delta(y_1, \dots,y_L) <0  \implies S(2^{ny_1 }, \dots, 2^{n y_L}) \xrightarrow{n \to \infty} 0.
\end{align}
In particular, if $ \Delta(y_1, \dots,y_L) >0,$ then
\begin{align}\label{eq: lim 1/n log (1-S) = - Delta}
     \lim_{n \to \infty} \frac{1}{n} \log (1-S(2^{ny_1 }, \dots, 2^{n y_L})) = - \Delta(y_1, \dots,y_L) < 0
\end{align}
and if $\Delta(y_1, \dots,y_L) < 0$, then 
\begin{align}\label{eq: lim 1/n log (S) =  Delta}
    \lim_{n \to \infty} \frac{1}{n} \log S(2^{ny_1 }, \dots, 2^{n y_L}) = \Delta(y_1, \dots,y_L)   < 0.
\end{align}
\end{proposition}
\begin{proof}
First note that as $R<1$ and $y_L < H_0 \leq 1$, for $n$ large enough, the condition $2^{\lfloor n y_L \rfloor} \leq 2^n - 2^{\lfloor nR \rfloor} + L$  of Proposition \ref{prop: exact SO proposition} is always satisfied. Also $I^N(y_1)< \infty$ as $y_1 <H_0$, so that $P( Q_1 = 2^{\lfloor ny_1 \rfloor }, \dots, Q_L = 2^{\lfloor ny_L \rfloor })>0$ for large enough $n$, cf. eq. \eqref{eq: P(Q_1=..) marginal probability}.  By Proposition \ref{prop: exact SO proposition} it follows that for large enough $n$, 
    \begin{align*}
   &S(2^{ny_1 }, \dots, 2^{n y_L}) \\
   &=\frac{ p_{N^n}(z^{n,2^{\lfloor ny_1 \rfloor }})}{\sum_{i=1}^{L}p_{N^n}(z^{n,2^{\lfloor ny_i \rfloor }})+ \frac{2^{\lfloor nR \rfloor} - L}{2^n - 2^{\lfloor ny_L \rfloor }}\sum_{i=2^{\lfloor ny_L \rfloor } + 1}^{2^n}p_{N^n}(z^{n,i})}.
    \end{align*}
   Let  $L\geq 2$ and $\Delta(y_1, \dots,y_L) > 0$.  Note that $1/n \log \left( \frac{2^{\lfloor nR \rfloor} - L}{2^n - 2^{\lfloor ny_L \rfloor }}\right) \xrightarrow{n \to \infty} -(1-R)$.  By Lemma \ref{lem: main lemma} and the principle of the largest term \cite[Lemma 1.2.15]{dembo2009large}, 
      \begin{align}
        &\lim_{n \to \infty} \frac{1}{n} \log \left(1- S(2^{ny_1 }, \dots, 2^{n y_L}) \right) \nonumber\\
        &=  - \min \left\{ \min_{i=2, \dots, L} \{ I^N(y_i) +y_i\} , 1-R + \inf_{x \geq y_L} I^N(x)  \right\} \label{eq: min delta >0 1} \\
        &+ \min \left\{ \min_{i=1, \dots, L} \{ I^N(y_i) +y_i\} , 1-R + \inf_{x \geq y_L} I^N(x)  \right\}. \label{eq: min delta >0 2}
   \end{align}
   By eq. \eqref{eq: limit P(1/n log G equal y)} and Assumption \ref{assumption: ML optimal guessing order assumption}, $y \mapsto I^N(y)+y$ is non-decreasing. Therefore, the
   minima in eq. \eqref{eq: min delta >0 1} and eq. \eqref{eq: min delta >0 2} reduce to 
   \begin{align*}
        &- \min \left\{  I^N(y_2) +y_2 , 1-R + \inf_{x \geq y_L} I^N(x)  \right\} \\
        &\qquad+ \min \left\{ I^N(y_1) +y_1 , 1-R + \inf_{x \geq y_L} I^N(x)  \right\}. 
   \end{align*}
From $\Delta(y_1, \dots,y_L) > 0$ it follows that $\min \left\{ I^N(y_1) +y_1 , 1-R + \inf_{x \geq y_L} I^N(x)  \right\} = I^N(y_1) +y_1$. This shows eq. \eqref{eq: lim 1/n log (1-S) = - Delta}.

To show  eq. \eqref{eq: lim 1/n log (S) =  Delta}, let $\Delta(y_1, \dots,y_L) < 0$. As $y \mapsto I^N(y)+y$ is non-decreasing and $\Delta(y_1, \dots,y_L) < 0$, $\Delta(y_1, \dots,y_L) $ simplifies to $
    \Delta(y_1, \dots,y_L) 
   =   1-R + \inf_{x \geq y_L} I^N(x)    - (I^N(y_1) +y_1 )
$ 
and we obtain
   \begin{align*}
        &\lim_{n \to \infty} \frac{1}{n} \log S(2^{ny_1 }, \dots, 2^{n y_L}) \\
        &=  - (I^N(y_1) +y_1 )  \\
        &\qquad + \min \left\{ \min_{i=1, \dots, L} \{ I^N(y_i) +y_i\} , 1-R + \inf_{x \geq y_L} I^N(x)  \right\} \\
        &=  - (I^N(y_1) +y_1 )  \\
        &\qquad + \min \left\{  I^N(y_1) +y_1 , 1-R + \inf_{x \geq y_L} I^N(x)  \right\} \\
        &=  - (I^N(y_1) +y_1 )  + 1-R + \inf_{x \geq y_L} I^N(x)  \\
        &=  \Delta(y_1, \dots,y_L).
   \end{align*}
   The case $L=1$ is analogous, and eqs. \eqref{eq: lim 1/n log (1-S) = - Delta} and \eqref{eq: lim 1/n log (S) =  Delta} imply eq. \eqref{eq: lim 1/n log (S) ->1} and eq. \eqref{eq: lim 1/n log (1-S) ->0}.
\end{proof}
By Proposition \ref{prop: main proposition}, the asymptotic reliability of the first codeword in the list depends only on the first, second and last likelihood ranks.
Moreover, we saw in the proof of Proposition \ref{prop: main proposition} that if $\Delta <0$, then $\Delta$ depends only on $y_1$ and $y_L$. Note that it also holds that $1-R + \inf_{x \geq y_L} I^N(x) < I^N(y_1)+y_1$ implies $\Delta(y_1, \dots, y_L) <0$. While $y \mapsto I^N(y)+y$ is non-decreasing, we will see in the proof of Corollary \ref{cor: L=1 threshold y*} that for $y> \gamma$, $y \mapsto I^N(y)+y$ is even strictly increasing, so that $ y_1 < y_2$ and $ \gamma <y_2$ imply  $I^N(y_1)+y_1 < I^N(y_2)+y_2$, and  $1-R + \inf_{x \geq y_L} I^N(x) > I^N(y_1)+y_1$ implies $\Delta(y_1, \dots,y_L) >0$. This shows the following.
\begin{corollary}\label{cor: if y_1 > gamma, then asymmptotic reliabiblity depends only on y_1,y_L}
    In the context of Proposition \ref{prop: main proposition}, let  $L \geq 2$ and $y_2 > \gamma$. Then 
    \begin{align*}
        1-R + \inf_{x \geq y_L} &I^N(x) > I^N(y_1)+y_1 \\
        &\implies S(2^{ny_1 }, \dots, 2^{n y_L}) \xrightarrow{n \to \infty} 1
    \end{align*}
    and 
    \begin{align*}
        1-R + \inf_{x \geq y_L} &I^N(x) < I^N(y_1)+y_1 \\
        &\implies S(2^{ny_1 }, \dots, 2^{n y_L}) \xrightarrow{n \to \infty} 0.
    \end{align*}
\end{corollary}
By Corollary \ref{cor: if y_1 > gamma, then asymmptotic reliabiblity depends only on y_1,y_L}, whenever $y_2 > \gamma$, whether the first codeword in the list is asymptotically reliable depends only on the first and last ranks. However, the rate of convergence in the case $ 1-R + \inf_{x \geq y_L} I^N(x) > I^N(y_1)+y_1$ still depends on $y_1,y_2$ and $y_L$ through $\Delta$ as defined in Proposition \ref{prop: main proposition}. 

Next, we recover the results of \cite[Theorem 3]{duffy_capacity-achieving_2019}. 
\begin{corollary}[Supercritical Guessing Threshold for $L=1$]\label{cor: L=1 threshold y*}
    Let $L=1$, and assume $H_{\min} < 1-R <H_0$, and $y_1 \in [0,H_0).$ Then there exists a unique $y^*$ characterized by $\Delta(y^*) = 0$ such that $S(2^{ny_1}) \xrightarrow{n \to \infty} 1 $ for $y_1 < y^*$ and $S(2^{ny_1}) \xrightarrow{n \to \infty} 0 $ for $y_1 > y^*$.

    If $R< 1-H$, the below-capacity case, then $y^* =1-R$, and if $R>1-H$, the above-capacity case, then $y^* < 1-R$.
\end{corollary}
\begin{proof}
By Theorem \ref{theo: G LDP}, $I^N$ is convex and has a unique root at $y=H$. Therefore, $I^N(y)$ is non-increasing for $y \in [0,H)$ and non-decreasing for $y \in [H,1).$
Therefore, $\inf_{x \geq y} I^N(x) = 0$ for $y \in [0,H]$, and $\inf_{x \geq y} I^N(x)= I^N(y)$ for $y \in (H,1]$. Thus, 
\begin{align}\label{eq: cases lambda(y1)}
    \Delta(y) =  \begin{cases}
    1-R -( I^N(y) + y), & y \in [0,H] \\
     1-R-y, & y \in [H,H_0).
\end{cases} 
\end{align} 
   As $H_{\min} < 1-R <H_0$, we get $\Delta(0) = 1-R - H_{\min} >0$ and $\Delta(y) \searrow 1-R-H_0 <0$ as $y \nearrow H_0.$ 
    The existence of $y^*$ then follows from the intermediate value theorem. 
   
   To show uniqueness, note that by Theorem \ref{theo: G LDP}, $I^N(y)+y= H_{\min}$ for $y \in [0,\gamma]$. So $\Delta(y)= 1-R-H_{\min}>0$ for $y \in [0,\gamma] $ and  it suffices to show that $y \mapsto \Delta(y)$ is strictly decreasing in $(\gamma,H_0).$  
    By eq. \eqref{eq: limit P(1/n log G equal y)} and Assumption \ref{assumption: ML optimal guessing order assumption}, $y \mapsto I^N(y)+y$ is non-decreasing. To show that $I^N(y)+y$ is strictly increasing in $(\gamma,H_0)$, suppose by way of contradiction that $I^N(y)+y$ were increasing, but not strictly increasing. Then  $I^N(y)+y \eqcolon c \in \mathbb{R}$ would be constant on some subinterval of $(\gamma,H_0).$  Then $I^N(y) = c-y$, which would contradict the strict convexity of $I^N(y)$. Therefore, $I^N(y)+y$ is strictly increasing on $(\gamma,H_0)$, which means $\Delta(y)$ is strictly decreasing on $(\gamma, H)$. For $y \in (H,H_0)$, $\Delta(y) = 1-R-y$, which is clearly strictly decreasing. 

    This shows the first part of the statement. For the second part, assume first that $R<1-H \iff H<1-R$. Then eq. \eqref{eq: cases lambda(y1)} gives $\Delta(1-R) = 1-R - (1-R) = 0$. As $y^*$ is the unique root of $\Delta$, it follows that $y^* = 1-R$. Next, assume $R>1-H \iff H>1-R$. Then $1-R \notin [H,H_0)$, so that the unique root of $\Delta$ must lie in $[0,H)$. 
    As $I^N$ has its unique root at $H$ and $y^* \in [0,H)$, $I^N(y^*) >0$, so that 
    $0= \Delta(y^*) = 1-R - (I^N(y^*)+y^*)$, which is equivalent to $y^* = 1-R -I^N(y^*) < 1-R. $
\end{proof}
From the proof of Corollary \ref{cor: L=1 threshold y*} we see that the assumption $H_{\min} < 1-R$ is necessary for any concentration on a correct decoding to occur at all. If $H_{\min} > 1-R$, then $\Delta(0) < 0 $. And as $y \mapsto \Delta(y)$ is non-increasing, it would follow that $\Delta(y) <0$ for all $y \in [0, H_0),$ which in turn would imply that $S(2^{ny}) \to 0$ for all such $y$. From Example \ref{ex: BSC gamma, H_0}, it follows that in the \ac{BSC}, $H_{\min} < 1-R$ is equivalent to $p<1-2^{-(1-R)}$, giving an upper bound on the noise level at which concentration on a correct decoding can occur. For example, if $R=0.7$, as in Fig. \ref{fig:finite n SO}, then $p_{\text{max}} \coloneqq 1-2^{-(1-0.7)}  \approx 0.1877$. On the other hand, the unique $p_{C} \in (0,1/2)$ at which $R = 1-H$ is $p_C \approx 0.0532$. So for $R=0.7$, $p<0.0532$ means the channel is below capacity, and $0.0532 <p<0.1877$ means the channel operates above capacity, but an early likelihood rank still guarantees concentration on a correct decoding. Only if $p>0.1877$, any decoding is likely the wrong one regardless of its corresponding likelihood rank.

We now move to the case $L=2$ and show that the second likelihood rank $Q_2$ carries valuable reliability information unobtainable from the first rank alone.
\begin{corollary}[The possible $L=2$ reliability gain]\label{cor: L=2 corollary}
    Let $L=2$, $H_{\min} < 1-R < H_0$,  and $y_1,y_2 \in [0,H_0)$ such that $y_1 < y_2$ and let $\gamma$ be as in Theorem \ref{theo: G LDP}. Let $y^*$ be defined as in Corollary \ref{cor: L=1 threshold y*}.
 \begin{enumerate}
     \item            Then, if $y_1 < y^*$, it holds that $S(2^{ny_1},2^{ny_2}) \xrightarrow{n \to \infty} 1 $ for all $y_2 \in (\max \{ y_1, \gamma\},H_0)$.
     \item    Assume in addition that   $H < H_0$. Set 
     \begin{align*}
         &\bar{y} \coloneqq \\
         &\sup\bigg\{  y \in [0, H_0) \;|\; y+I^N(y) < 1-R + I^N(H_0) \bigg\}.
     \end{align*}
     Then $\bar{y}>y^*$ and for all $y_1 \in (y^*,\bar{y})$ there exists a unique $\delta >0$, such that $y_1 + \delta \in (H,H_0)$, $\Delta(y_1, y_1+\delta)=0$, and for all $y_2 \in (y_1,y_1+\delta)$ it holds that $S(2^{ny_1},2^{ny_2}) \xrightarrow{n \to \infty} 0 $ and for all $y_2 \in (y_1+\delta,H_0)$ it holds that $S(2^{ny_1},2^{ny_2}) \xrightarrow{n \to \infty} 1 $. 
     \item            Furthermore, if $I^N(H_0) = \infty,$ then $\bar{y} = H_0$, but if $I^N(H_0) < \infty,$ then $\bar{y}<H_0$ and it holds for all $y_1 \in (\bar{y},H_0)$ and for all $y_2 \in (y_1, H_0)$ that $S(2^{ny_1},2^{ny_2}) \xrightarrow{n \to \infty} 0 $. 
 \end{enumerate}
 All of the above convergences have exponential rates as specified in Proposition \ref{prop: main proposition}.
\end{corollary}
\begin{proof}
    Owing to Lemma \ref{lem: I(H_0)}, we use $I^N(H_0)$ and $ \lim_{y \nearrow H_0} I^N(y)$ interchangeably throughout the proof.
    
    To show the first part, let $y_1 < y^*$. By Corollary \ref{cor: L=1 threshold y*}, it follows that $\Delta(y_1)>0$, i.e., $1-R + \inf_{x \geq y_1} I^N(x) > I^N(y_1) +y_1.$ As $\inf_{x \geq y_2} I^N(x) \geq \inf_{x \geq y_1} I^N(x),$ it follows that also $1-R + \inf_{x \geq y_2} I^N(x) > I^N(y_1) +y_1.$ So in order to show $\Delta(y_1,y_2)>0$ for $y_2 \in (\max \{ y_1, \gamma\},H_0)$, it suffices to show that $I^N(y_2) + y_2 > I^N(y_1) + y_1$ for such $y_2$. But this is true, as we know from the proof of Corollary \ref{cor: L=1 threshold y*} that $y \mapsto I^N(y)+y$ is strictly increasing in $(\gamma,H_0)$.

    To show the second and third parts,
    note first that since $H<H_0$,  $\inf_{x \geq y}I^N(x) = I^N(y)$ for $y\geq H$, so that $\lim_{y \nearrow H_0} \inf_{x \geq y}I^N(x) = \lim_{y \nearrow H_0} I^N(y) = I^N(H_0).$ 

    As  $H<H_0$ and $I^N(y)$ is strictly increasing for $y \in (H,H_0)$, we have that $I^N(y^*)+y^* = 1-R +\inf_{x \geq y^*}I^N(x)<1-R+ \lim_{y \nearrow H_0} \inf_{x \geq y}I^N(x)$.  This shows $\bar{y}>y^*.$  From the proof of Corollary \ref{cor: L=1 threshold y*} we know that  $\Delta(\gamma) >0 $, so in particular $ y^* > \gamma$. As $y \mapsto I^N(y)+y$ is strictly increasing in $(\gamma,H_0)$, it follows that $ y_2 + I^N(y_2) > y_1 + I^N(y_1) $ for all $y_2 >y_1>y^*$ and  for such $y_1, y_2$, the sign of $\Delta (y_1,y_2)$ is entirely determined by whether or not  $1-R +   \inf_{x \geq y_2}I^N(x)  > y_1 + I^N(y_1)$.

    Let $y_1 \in (y^*,\bar{y}).$ Then in particular, $y_1$ satisfies 
    \begin{align}
      1-R +   \inf_{x \geq y_1}I^N(x)  &< y_1 + I^N(y_1) \label{eq: condition y_1}\\
      &< 1-R + \lim_{y \nearrow H_0} \inf_{x \geq y}I^N(x). \label{eq: condition y_1 2}
    \end{align}
       As $\inf_{x \geq y} I^N(x) = 0$ for $y \in [0,H]$, and $\inf_{x \geq y} I^N(x)= I^N(y)$ for $y \in (H,1]$,  the left-hand side of eq. \eqref{eq: condition y_1} is flat and equals $1-R$ for $y_1 \in [0,H]$ and is strictly increasing and continuous thereafter, and converges to the right-hand side  of eq. \eqref{eq: condition y_1 2} as $y_1 \nearrow H_0$. By the intermediate value theorem, there exists a unique $\delta >0$, such that $\delta +y_1 > H$, and 
$ 1-R +   \inf_{x \geq y_2}I^N(x)  < y_1 + I^N(y_1)$ for all $y_2 \in (y_1,y_1 + \delta)$, and $ 1-R +   \inf_{x \geq y_2}I^N(x)  > y_1 + I^N(y_1)$ for all $y_2 \in (y_1 + \delta, H_0)$. Consequently, it holds for all $y_2 \in (y_1,y_1 + \delta)$ that $\Delta(y_1, y_2) <0$, and for all $y_2 \in (y_1 + \delta, H_0)$ that 
     $\Delta(y_1, y_2) >0$. This shows the second part.

    To show the third part, note that clearly $\bar{y} = H_0$ whenever $I^N(H_0) = \infty$.
    So let $I^N(H_0) < \infty$. Then 
    $y+ I^N(y) \nearrow H_0 + I^N(H_0) > 1-R + I^N(H_0)$ as $y \nearrow H_0$, which shows $\bar{y}< H_0$.
    Let $y_1 \in (\bar{y},H_0)$. Then 
    \begin{align*}
        y_1 + I^N(y_1) 
      &> 1-R + \lim_{y \nearrow H_0} \inf_{x \geq y}I^N(x) \\
      &\geq 1-R + \inf_{x \geq y_2}I^N(x)
    \end{align*}
  for all $y_2 \in (y_1,H_0).$  This implies $\Delta(y_1,y_2)<0$ for those $y_1,y_2$.
\end{proof}
The results of Corollary \ref{cor: L=2 corollary} are illustrated for the \ac{BSC} in Fig. \ref{fig:phase_diagram}.  By Example \ref{ex: BSC gamma, H_0} we know that in the \ac{BSC}, $I^N(H_0)< \infty$, so that $\bar{y}<H_0$ and, as predicted by Corollary \ref{cor: L=2 corollary}, we obtain the three intervals $[0,y^*)$, $(y^*, \bar{y})$ and $(\bar{y},1)$.  If $y_1$, the normalized logarithm of the first likelihood rank, lies in $[0,y^*)$, then the position of the second rank is irrelevant and concentration of the first list entry on a correct decoding is guaranteed. If $y_1 \in (y^*, \bar{y})$, then a sufficiently late $y_2$ rescues confidence in the decoding. This is in contrast to the $L=1$ case, where $y_1 > y^*$ implies concentration on a wrong decoding. However, even in the $L=2$ case, a sufficiently late first rank,  $y_1 \in (\bar{y},1)$, again implies that the position of the second rank is irrelevant and concentration on a wrong decoding is guaranteed. 

We will now study the explicit relationship between $p$, the crossover probability of the \ac{BSC}, the code rate $R$, and the objects introduced in the theory developed above. Our first goal is to obtain a closed-form expression for $I^N(y)+y$, which repeatedly appears throughout Proposition \ref{prop: main proposition} and Corollaries \ref{cor: L=1 threshold y*} and \ref{cor: L=2 corollary}. 

As the \ac{BSC} is memoryless, apart from the characterization of the guesswork rate function $I^N$ given in \cite{christiansen2013_guesswork}, which is based on the Legendre-Fenchel transform of the scaled cumulant generating function of $1/n \log G(N^n)$,  one can use tilted distributions and typical-set arguments to obtain an alternative characterization of $I^N$  \cite{Beirami2019}. Let $\boldsymbol{\mu}$ be the distribution on $\fb_2$ associated with the \ac{BSC}, i.e., $\boldsymbol{\mu}(1)= p$ and $\boldsymbol{\mu}(0)= 1-p$. For tilt parameter $\alpha \in \mathbb{R},$ let $T(\boldsymbol{\mu},\alpha)$ be the tilted $\boldsymbol{\mu}$  of order $\alpha$, given by 
\begin{align*}
   T(\boldsymbol{\mu},\alpha)(1)= \frac{p^{\alpha}}{p^{\alpha} + (1-p)^{\alpha}}, \; T(\boldsymbol{\mu},\alpha)(0)= \frac{(1-p)^{\alpha}}{p^{\alpha} + (1-p)^{\alpha}}.
\end{align*}
Then \cite[Theorem 5]{Beirami2019} yields 
\begin{align*}
    I^N(y) = D(T(\boldsymbol{\mu},\alpha(y)) \Vert \boldsymbol{\mu}),
\end{align*}
where 
\begin{align}\label{eq: alpha(y) definition}
    \alpha (y) = \arg_{\alpha >0} \{ H( T(\boldsymbol{\mu},\alpha))=y \}.
\end{align}
Here, for distributions $\boldsymbol{\rho}_1, \boldsymbol{\rho}_2$ on $\fb_2,$  $H(\boldsymbol{\rho}_1) = - \boldsymbol{\rho}_1(0)\log \boldsymbol{\rho}_1(0) - \boldsymbol{\rho}_1(1) \log \boldsymbol{\rho}_1(1)  $ is the Shannon entropy, and $D (\boldsymbol{\rho}_1 \Vert \boldsymbol{\rho}_2) = \boldsymbol{\rho}_1(0) \log(\boldsymbol{\rho}_1(0)/\boldsymbol{\rho}_2(0))+ \boldsymbol{\rho}_1(1) \log(\boldsymbol{\rho}_1(1)/\boldsymbol{\rho}_2(1))$ is the relative entropy. Let
\begin{align*}
     h \colon [0,1] &\to [0,1]\\
     y &\mapsto  - y \log(y)- (1-y) \log (1-y)
 \end{align*}
 be the binary entropy function, extended to $y=0$ and $y=1$ by continuity. With this notation, we obtain $H( T(\boldsymbol{\mu},\alpha)) = h \left( \frac{p^{\alpha}}{p^{\alpha} + (1-p)^{\alpha}} \right).$ And since $\alpha \mapsto \frac{p^{\alpha}}{p^{\alpha} + (1-p)^{\alpha}}$ is a bijection from $(0, \infty)$ to $(0, 1/2),$ and $h_{\vert_{[0,1/2]}}$ is also a bijection, 
we obtain that the unique solution to eq. \eqref{eq: alpha(y) definition} satisfies $   \frac{p^{\alpha(y)}}{p^{\alpha(y)} + (1-p)^{\alpha(y)}} = h^{-1}(y)$, where $h^{-1}$ is the inverse of $h_{\vert_{[0,1/2]}}$. Let  
 \begin{align*}
     d(w_1 \Vert w_2) = w_1 \log \left( \frac{w_1}{w_2} \right) + (1- w_1) \log \left( \frac{1-w_1}{1-w_2} \right)
 \end{align*}
 be the binary relative entropy function for $0<w_1,w_2 <1$. Then
$ D(T(\boldsymbol{\mu},\alpha) \Vert \boldsymbol{\mu}) = d\left( \frac{p^{\alpha}}{p^{\alpha} + (1-p)^{\alpha}} \Vert p\right)$ and we obtain that $I^N(y) = d\left( h^{-1} (y)\Vert p\right)$. Further simplification yields the following result.
\begin{example}[\acs{BSC} continued]
    In the \ac{BSC} with $0<p<\frac{1}{2}$, let $h^{-1} \colon [0,1] \rightarrow [0, \frac{1}{2}]$ be the inverse of $h_{\vert_{[0,1/2]}}$. Then for all $y \in [0,1]$ it holds that
     \begin{align*}
          I^N(y) + y &=  h(h^{-1}(y) \Vert p) \\
          &\coloneqq - h^{-1}(y) \log(p) - \left(1- h^{-1}(y)\right) \log(1-p),
      \end{align*}
      where $h(h^{-1}(y) \Vert p)$ is the binary cross-entropy of $h^{-1}(y)$ relative to $p$. 
      
     Following eq. \eqref{eq: cases lambda(y1)}, we obtain that for $L=1$, 
      \begin{align*}
    \Delta(y) =  \begin{cases}
    1-R -h(h^{-1}(y) \Vert p), & y \in [0,h(p)] \\
     1-R-y, & y \in [h(p),1),
\end{cases} 
\end{align*} 
and similarly for $L>1$
\begin{align*}
  &\Delta(y_1,\dots,y_L)\\
  &=\begin{cases}
      \begin{aligned}[t]
        \min\bigl\{\,& h(h^{-1}(y_2) \Vert p),
                     1-R\bigr\} \\
                     &-h(h^{-1}(y_1) \Vert p),
      \end{aligned}
      & y_L \in[0,h(p)],\\[3ex]
      \begin{aligned}[t]
        \min\bigl\{\,& h(h^{-1}(y_2) \Vert p),\\
                     & 1-R+ d\left( h^{-1} (y_L)\Vert p\right)\bigr\}\\
                     &\;-h(h^{-1}(y_1) \Vert p),
      \end{aligned}
      & y_L\in[h(p),1).
    \end{cases}
\end{align*}
In the below-capacity case $R < 1-h(p)$, it follows that $y^* = 1-R$, and in the above-capacity case  $R > 1-h(p)$, we obtain 
\begin{align}\label{eq: y^* in the BSC}
    y^* = h \left( \frac{1-R + \log(1-p)}{ \log \left(  \frac{1-p}{p}\right)}\right).
\end{align}
Notice that $H_{\min} < 1-R$ guarantees $y^* >0$. $H=h(p) < H_0=1$ is satisfied as $0<p<\frac{1}{2}$. Also $I^N(H_0)= - \log\left(2 \sqrt{p(1-p)}\right)< \infty$, so that we obtain 
\begin{align}\label{eq: y_bar in the BSC}
    \bar{y} = h \left(  \frac{1}{2} - \frac{R}{ \log \left( \frac{1-p}{p} \right)}\right).
\end{align}
\end{example}
With this, we can precisely quantify the numerical values for $y^*$ and $ \bar{y}$, given in Fig. \ref{fig:phase_diagram}, where $R=0.7$ and $p=0.1$, as 
\begin{align*}
    y^* = h \left( \frac{0.3 + \log(0.9)}{ \log \left(  9\right)}\right) \approx 0.272
\end{align*}
and 
\begin{align*}
    \bar{y} = h \left(  \frac{1}{2} - \frac{0.7}{ \log \left( 9 \right)}\right) \approx 0.854.
\end{align*}
To close this section, we remark that in the \ac{BSC}, Assumption \ref{assumption: ML optimal guessing order assumption} corresponds to an ordering by increasing Hamming weight, breaking ties arbitrarily. Therefore, it is natural to expect that the likelihood-rank thresholds $y^*$ and $\bar{y}$ have corresponding Hamming weight equivalents. To identify them, notice that, using elementary estimates linking the volume of a Hamming ball to the binary entropy function, cf. \cite[Example 11.1.3]{Cover2006},  one can show that for all $y \in  [0,1)$ it holds that
    \begin{align*}
        \lim_{n \to \infty} h \left(\frac{\wt (z^{n,2^{\lfloor n y \rfloor}})}{n} \right) = y \text{ and } \frac{\wt (z^{n,2^{\lfloor n y \rfloor}})}{n} \leq \frac{1}{2}
    \end{align*}
    for all $n \in \mathbb{N}.$
Comparing this with eq. \eqref{eq: y^* in the BSC} and eq. \eqref{eq: y_bar in the BSC} shows that, in the above-capacity regime $R> 1-h(p)$, 
\begin{align*}
    w^* \coloneqq \frac{1-R + \log(1-p)}{ \log \left(  \frac{1-p}{p}\right)} \text{ and } \bar{w} \coloneqq  \frac{1}{2} - \frac{R}{ \log \left( \frac{1-p}{p} \right)}
\end{align*}
represent normalized Hamming weight thresholds. Below capacity, $R< 1-h(p)$,  $\bar{w}$ stays the same and $w^* = h^{-1}(1-R).$ It follows that  $ \lim_{n \to \infty}  \wt (z^{n,2^{\lfloor n y_1 \rfloor}})/n =  h^{-1}(y_1) \eqqcolon w_1$. As $h_{\vert_{[0,1/2]}}$ is strictly increasing, so is $h^{-1}$ and it holds that $w_1 < w^*$ if and only if $y_1 < y^*$,  $ w^* <w_1 < \bar{w}$ if and only if  
$ y^* <y_1 < \bar{y}$, and $w_1> \bar{w}$ if and only if $y_1 > \bar{y}$.
Therefore, Corollary \ref{cor: L=2 corollary} yields that if $w_1 < w^*$, concentration on a correct decoding is guaranteed. If $ w^* <w_1 < \bar{w}$, the rescue phenomenon can occur, but if $w_1> \bar{w}$, the decoding will concentrate on a wrong result. In particular, notice that in the above-capacity regime, the size of the rescue window
\begin{align*}
    \bar{w} - w^* = \frac{1}{2}  - \frac{1 + \log(1-p)}{ \log \left(  \frac{1-p}{p}\right)}
\end{align*}
is independent of the code rate $R$.

\section{Binary Linear Codes}\label{sec: Binary Linear Codes}
\begin{figure*}[t]
    \centering
    \includegraphics[width=\textwidth]{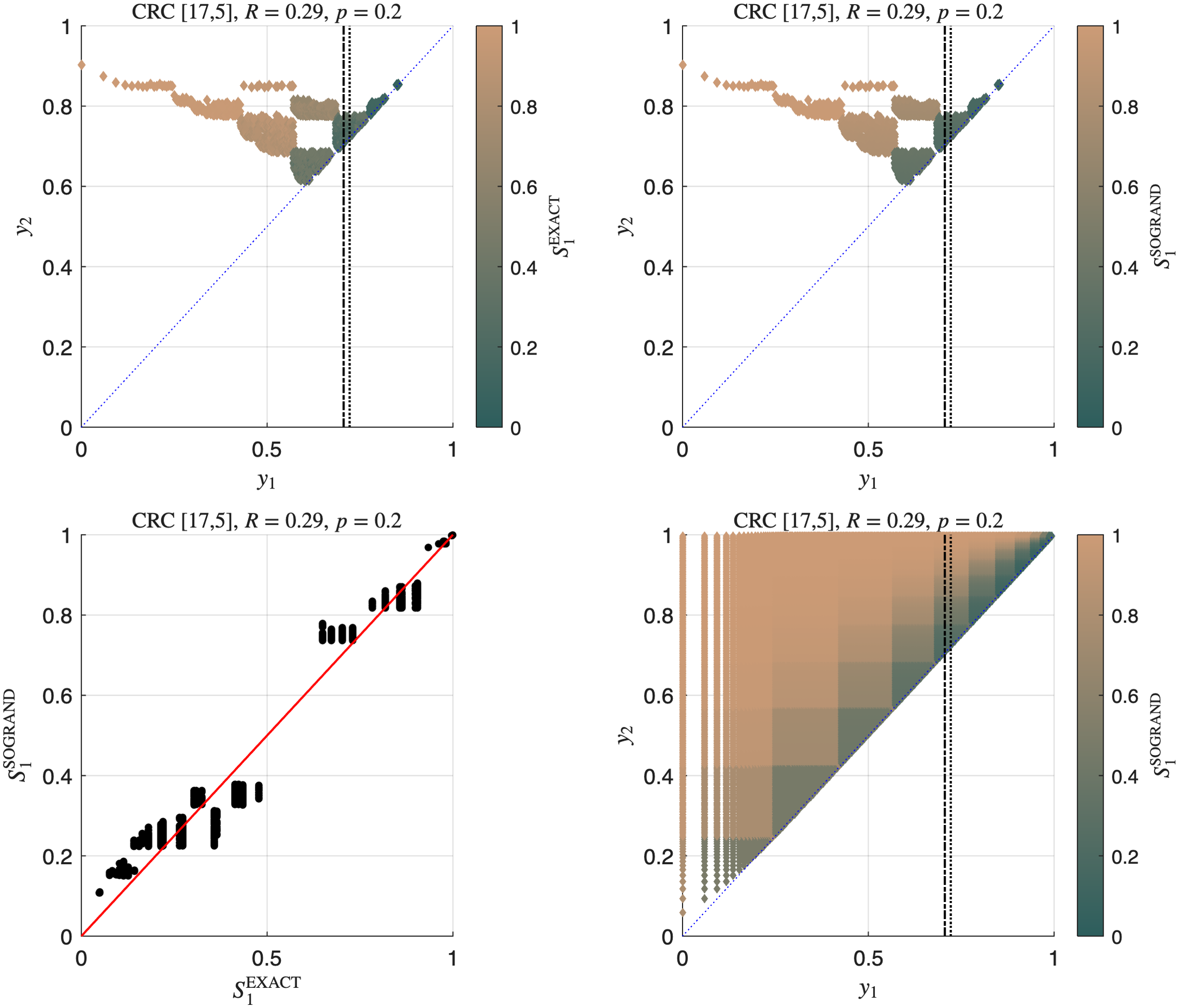}
    \caption{For a CRC[17,5] code in the \ac{BSC}, we show a comparison between the exact $S_1^{\text{EXACT}}(s) \coloneqq P(X^n =  \widehat{C}^{n,1} \;|\; Q_1 = q_1 (s),  Q_2 = q_2(s))$ as specified by eq. \eqref{eq: SO fraction, binary linear code} in Lemma \ref{lem: structured code SO} and the \ac{SOGRAND} approximation $S_1^{\text{SOGRAND}}(s)  $ as given by eq. \eqref{eq: SO fraction 2}. The parameters are $R=0.29$ and $p=0.2$. The top left shows $S_1^{\text{EXACT}}$  evaluated for all possible pairs $(q_1(s),q_2(s))$ as specified in Lemma \ref{lem: syndrome HN^n determines Q_i}. The top right shows $S_1^{\text{SOGRAND}}(s)$ for the same pairs.  The bottom left shows a scatter plot comparing $S_1^{\text{EXACT}}(s)$ to $S_1^{\text{SOGRAND}}(s)$ for  all $(q_1(s),q_2(s))$ pairs, and the bottom right shows $S_1^{\text{SOGRAND}}$ evaluated for all pairs $(q_1,q_2) \in [2^n]^2$, $q_1 < q_2$, even those that the linear structure of the code forbids. The x- and y-axes are scaled as $1/n  \log$ to facilitate comparison with Fig. \ref{fig:finite n SO}. }
    \label{fig:SO comparison CRC}
\end{figure*}

\begin{figure*}[t]
    \centering
    \includegraphics[width=\textwidth]{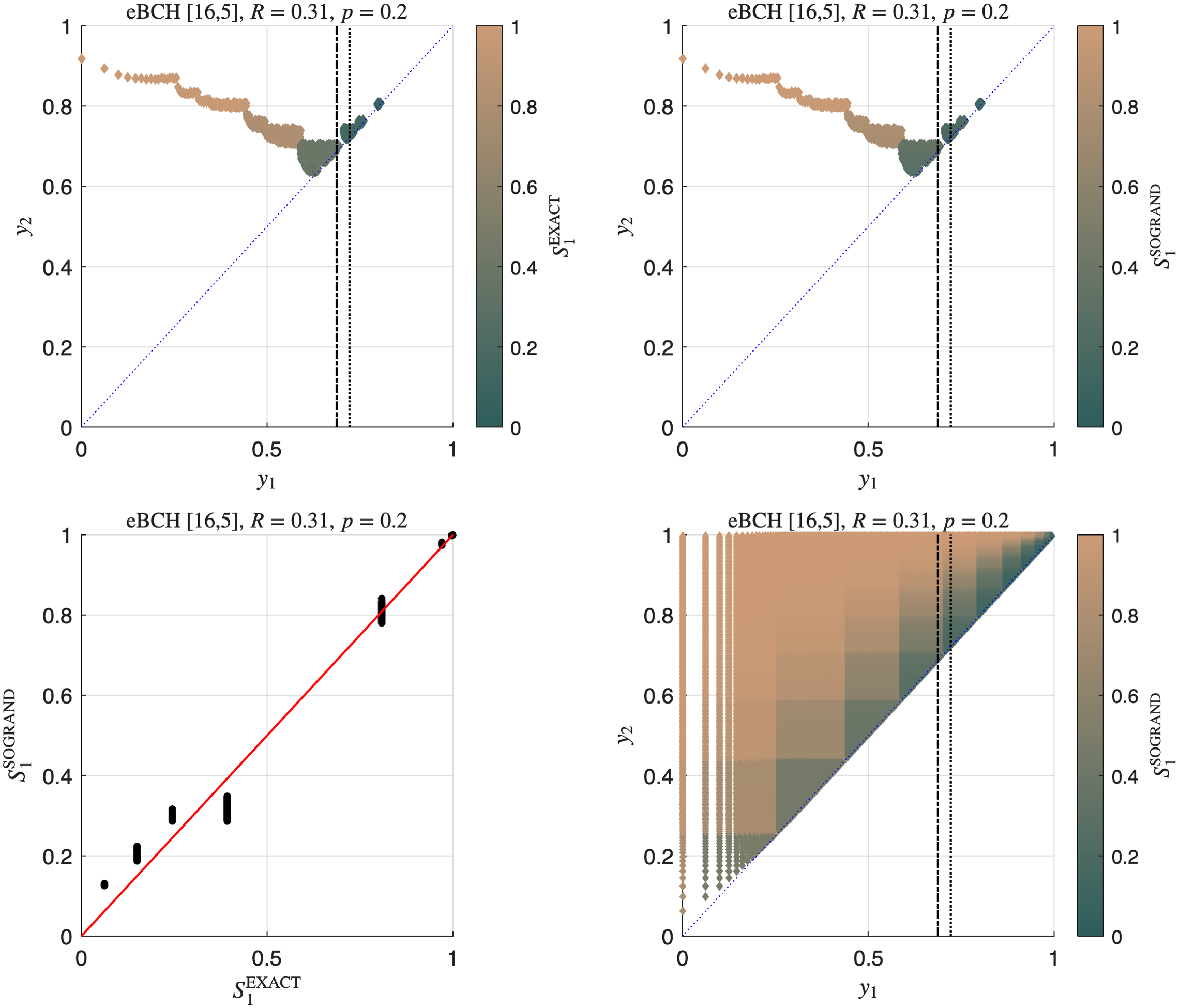}
    \caption{ For an eBCH[16,5] code in the \ac{BSC}, we show an analogous plot to Fig. \ref{fig:SO comparison CRC}. The parameters are $R=0.31$ and $p=0.2$. }
    \label{fig:SO comparison eBCH}
\end{figure*}
For small $k$, where the blockwise \ac{SO} from eq. \eqref{eq: definition blockwise so} can be evaluated through complete enumeration, we now illustrate that the phenomenon identified in the previous sections remains true for binary linear codes. Unless specified otherwise, we again assume that $G$ is any bijection, so that our finite-blocklength results again hold even for approximations of \ac{ML} decoding. To do that, we assume that the codebook $\mathcal{C}_n$ is no longer random, but instead a binary linear code, cf. \cite[Chapter 3]{lin_error_2004}. In particular, let $\mathcal{C}_n$ be a $k$-dimensional subspace of $\fb_2^n$ that is characterized as the kernel of its parity-check matrix $\mathbf{H} \in \fb_2^{(n-k) \times n}$, where $\rank \mathbf{H} = n-k$. Then 
\begin{align}\label{eq: definition linear code via parity check matrix}
    \mathcal{C}_n \coloneqq \ker \mathbf{H} = \{ c^n \in \fb_2^n \; | \; \mathbf{H}c^n = 0 \}.
\end{align}

We will first show that owing to the linearity of the code $\mathcal{C}_n$, the likelihood ranks $Q_1, \dots, Q_{L}$ do not depend on the transmitted codeword $X^n$ anymore. They depend only on the noise $N^n$ through its syndrome $\mathbf{H} N^n$. Even stronger, we show that once a noise realization $N^n = z^n$ is fixed, the entire sequence $Q_1, \dots, Q_{2^k}$ is a deterministic function of the syndrome $\mathbf{H}z^{n}$, regardless of whether in practice a list decoder stops at list size $L<2^k$.
\begin{lemma}\label{lem: syndrome HN^n determines Q_i}
    Let $\mathcal{C}_n$ be a binary linear code with parity-check matrix $\mathbf{H}$.  Then the ranks $Q_j$, $j \in [2^k]$, of the noise candidates corresponding to the list entries do not depend on the sent codeword $X^n$, but depend only on the noise $N^n$. In particular, they depend on $N^n$ only through its syndrome $\mathbf{H}N^n$  and can be written as $Q_j = q_j(\mathbf{H}N^n)$ with deterministic functions $q_j \colon \fb_2^{n-k} \rightarrow [2^n]  $, $j \in [2^k]$, that are defined for all $s \in \fb_2^{n-k}$ as 
    \begin{align}\label{eq: q_1(s) function}
        q_1(s)= \min \{ G(z^n) \; | \; z^n \in \fb_2^n \text{ s.t. } \mathbf{H}z^n = s \}
    \end{align}
    and 
        \begin{align}
        q_j(s)= \min \{ G(z^n) \; | \; z^n \in \fb_2^n \text{ s.t. } \mathbf{H}z^n = s, \nonumber
        \qquad \\ G(z^n) \notin \{ q_1(s), \dots, q_{j-1}(s) \} \} \label{eq: q_i(s) function},
    \end{align}
    for $j =2, \dots, 2^k.$
\end{lemma}
\begin{proof}
    The definitions of the $Q_j$ remain as in eq. \eqref{eq: second1 definition Q_j} and eq. \eqref{eq: second2 definition Q_j}. However, owing to the linearity of $\mathcal{C}_n$ we obtain
 \begin{align*}
    Q_1 &= \min_{C^n \in \mathcal{C}_n} G(Y^n - C^n ) = \min_{C^n \in \mathcal{C}_n} G((X^n-C^n)+N^n) \\
    &=  \min_{C^n \in \mathcal{C}_n} G(C^n+N^n ),
\end{align*}
where we used that as $C^n$ ranges over $\mathcal{C}_n$,  $X^n-C^n$ also ranges over $\mathcal{C}_n$. This shows that $Q_1$ is the minimum value $G$ takes on the coset $N^n + \mathcal{C}_n \coloneqq \{ N^n + C^n \; | \; C^n \in \mathcal{C}_n \}$. Analogously, $Q_j$ is the $j$-th smallest value $G$ takes on $N^n + \mathcal{C}_n$. 

To complete the proof, note that 
\begin{align}\label{eq: coset N+ C_n syndrome characterization}
    N^n + \mathcal{C}_n = \left\{ z^n \in \fb_2^n \;|\;  \mathbf{H}z^n = \mathbf{H} N^n  \right\}.
\end{align}
To see why this is true, first let $z^n = N^n + C^n \in N^n + \mathcal{C}_n.$ By eq. \eqref{eq: definition linear code via parity check matrix}, it follows that $\mathbf{H}z^n = \mathbf{H}N^n + \mathbf{H}C^n = \mathbf{H}N^n.$ For the other inclusion, let $z^n \in \fb_2^n $ be such that  $\mathbf{H}z^n = \mathbf{H} N^n$. Define $C^n \coloneqq z^n + N^n.$ Then it follows by linearity, and since in $\fb_2^n$, the sum of any two equal vectors is $0$, that $\mathbf{H}C^n = 0$. So by eq. \eqref{eq: definition linear code via parity check matrix}, $C^n \in \mathcal{C}_n$. Therefore, $z^n = C^n + N^n \in N^n + \mathcal{C}_n,$ which shows eq. \eqref{eq: coset N+ C_n syndrome characterization}. 
\end{proof}
Lemma \ref{lem: syndrome HN^n determines Q_i} also implies that only those ranks $q_1, \dots, q_L \in [2^n]$ are possible which can be represented by eqs. \eqref{eq: q_1(s) function}, \eqref{eq: q_i(s) function}.
By assumption, $\rank  \mathbf{H} = n-k$, so that every $s \in \fb_2^{n-k}$ is a valid syndrome, i.e., for every $s \in \fb_2^{n-k}$ there exists a $z^n \in \fb_2^n$ such that $s =\mathbf{H} z^n. $ 
The possible ranks are therefore 
\begin{align}\label{eq: all possible q_1, q_2 pairs}
     \{ (q_1 (s), \dots ,q_{2^k}(s)) \; | \; s \in \fb_2^{n-k} \},
\end{align}
with $q_1(s),  \dots ,q_{2^k}(s)$ being the deterministic functions in eqs. \eqref{eq: q_1(s) function}, \eqref{eq: q_i(s) function}.

The next Lemma \ref{lem: structured code SO} is the  binary linear code analogue of Proposition \ref{prop: exact SO proposition}, which said that for random codebooks, the \ac{SOGRAND} formula eq. \eqref{eq: SOGRAND formula introduction} equals eq. \eqref{eq: SO definition introduction}, the conditional probability of correctness, conditioned on the likelihood ranks $Q_1, \dots, Q_L$. In Lemma \ref{lem: syndrome HN^n determines Q_i} we saw that the noise syndrome $\mathbf{H}N^n$ already determines the full sequence $Q_1, \dots, Q_{L}$ of likelihood ranks. We will see that, for this reason, conditioning on the information available after $L=1$ decoding is equivalent to conditioning on the information available after finding all $L=2^k$ codewords. In particular, in the binary linear code case, evaluating the exact \ac{SO} requires, even after finding the first codeword, summing over the probability mass of all $2^k$ candidate noise sequences, or equivalently leveraging knowledge of the code structure to evaluate the syndrome probability $P(\mathbf{H}N^n =s)$.
\begin{lemma}\label{lem: structured code SO}
     Let $\mathcal{C}_n$ be a binary linear code with parity-check matrix $\mathbf{H}$, where $\rank \mathbf{H} = n-k$, and let $s \in \fb_2^{n-k}$ be a syndrome. Let $q_1 (s)<q_2 (s)< \dots <q_{2^k}(s)$ be the sequence of likelihood ranks determined by $s$ through Lemma \ref{lem: syndrome HN^n determines Q_i}. If $P(\mathbf{H}N^n = s)>0$,  then 
     \begin{align}
         &P(X^n =  \widehat{C}^{n,j} \;|\; Q_1 = q_1(s)) \label{eq: cond correctness L=1} \\
         &=  P(X^n =  \widehat{C}^{n,j} \;|\; Q_1 = q_1 (s), \dots, Q_L = q_L(s)) \label{eq: cond correctness L}\\
         &=  P(X^n =  \widehat{C}^{n,j} \;|\; Q_1 = q_1 (s), \dots, Q_{2^k} = q_{2^k}(s)) \label{eq: cond correctness L=2^k}\\
         &=  \frac{ p_{N^n}(z^{n,q_j(s)})}{P(\mathbf{H}N^n =s)} = \frac{ p_{N^n}(z^{n,q_j(s)})}{\sum_{i=1}^{2^k}p_{N^n}(z^{n,q_i (s)})}, \label{eq: SO fraction, binary linear code}
     \end{align}
     for any $L \in [2^k]$ and $j \in [L]$ and where $p_{N^n}$ is the probability mass function of $N^n$.
\end{lemma}
\begin{proof}
    We start by showing that the three conditional probabilities in eqs. \eqref{eq: cond correctness L=1}, \eqref{eq: cond correctness L} and \eqref{eq: cond correctness L=2^k} are equal. To do that, we first show the event equality
    \begin{align}\label{eq: event equality HN^n =s iff Q_1 = q_1(s)}
        \{ \mathbf{H}N^n = s \} = \{ Q_1 = q_1 (s)\}.
    \end{align}
    If $\mathbf{H}N^n = s$, then Lemma \ref{lem: syndrome HN^n determines Q_i} yields $Q_1= q_1(\mathbf{H}N^n) = q_1(s).$ To show the other inclusion, let $Q_1 = q_1(s).$ By definition, that means
    \begin{align}
       &\min \{ G(z^n) \; | \; z^n \in \fb_2^n \text{ s.t. } \mathbf{H}z^n = \mathbf{H}N^n \} \label{eq: set 1}\\
       =& \min \{ G(z^n) \; | \; z^n \in \fb_2^n \text{ s.t. } \mathbf{H}z^n = s \} \label{eq: set 2}.
    \end{align}
    As $G$ is a bijection, the respective minimizers of eqs. \eqref{eq: set 1} and \eqref{eq: set 2} are unique. Let $z^n_1$ be the minimizer of eq. \eqref{eq: set 1} and $z^n_2$ be the minimizer of eq. \eqref{eq: set 2}. Using again that $G$ is a bijection, it follows that
    \begin{align*}
        G(z_1^n) = G(z_2^n) \iff z^n_1 = z_2^n  &\implies \mathbf{H}z^n_1 = \mathbf{H}z^n_2 \\
        &\implies \mathbf{H}N^n  = s.
    \end{align*}
    This shows eq. \eqref{eq: event equality HN^n =s iff Q_1 = q_1(s)}. But by Lemma \ref{lem: syndrome HN^n determines Q_i}, the syndrome $\mathbf{H}N^n$ determines the full sequence $Q_1, \dots, Q_{2^k}$. So all three conditioning events in eqs. \eqref{eq: cond correctness L=1}, \eqref{eq: cond correctness L} and \eqref{eq: cond correctness L=2^k} are equal, and equal $  \{ \mathbf{H}N^n = s \}$. Therefore, the conditional probabilities in   eqs. \eqref{eq: cond correctness L=1}, \eqref{eq: cond correctness L} and \eqref{eq: cond correctness L=2^k} are  also equal.

    To conclude the proof, note that whenever $Q_1 = q_1(s),$
    \begin{align*}
        X^n =  \widehat{C}^{n,j} = Y^n - z^{n,q_j (s)} \iff N^n = z^{n,q_j (s)},
    \end{align*}
    so that with eq. \eqref{eq: event equality HN^n =s iff Q_1 = q_1(s)} it follows that 
    \begin{align}\label{eq: P(- | Q_1 = q_ ) = P(- | HN=s)}
        P(X^n =  \widehat{C}^{n,j} \;|\; Q_1 = q_1(s)) =  P( N^n = z^{n,q_j (s)} \;|\;  \mathbf{H}N^n = s).
    \end{align}
    By Lemma \ref{lem: syndrome HN^n determines Q_i},  it follows that $q_1(s), \dots, q_{2^k}(s)$ are the ranks of noise vectors $z^{n,q_1(s)}, \dots, z^{n,q_{2^k}(s)}$ all having syndrome $s$. At the same time eq. \eqref{eq: coset N+ C_n syndrome characterization} implies that there are $|\mathcal{C}_n|=2^k$ such vectors, so  $z^{n,q_1(s)}, \dots, z^{n,q_{2^k}(s)}$ are precisely the vectors in $\fb_2^n$ having syndrome $s$. From that and  eq. \eqref{eq: P(- | Q_1 = q_ ) = P(- | HN=s)}, equality holds throughout eqs. \eqref{eq: cond correctness L=1}, \eqref{eq: cond correctness L}, \eqref{eq: cond correctness L=2^k} and \eqref{eq: SO fraction, binary linear code}.
\end{proof}
In Figs. \ref{fig:SO comparison CRC} and \ref{fig:SO comparison eBCH} we evaluate for a CRC[17,5] code and an eBCH[16,5] code, respectively, both the exact conditional probability of correctness $S_1^{\text{EXACT}}(s) \coloneqq  P(X^n =  \widehat{C}^{n,1} \;|\; Q_1 = q_1 (s),  Q_2 = q_2(s))$ as specified by eq. \eqref{eq: SO fraction, binary linear code} in Lemma \ref{lem: structured code SO} and the \ac{SOGRAND} approximation $S_1^{\text{SOGRAND}}(s)  $ as given by eq. \eqref{eq: SO fraction 2} in Proposition \ref{prop: exact SO proposition}. We evaluate $S_1^{\text{EXACT}}(s)$ and  $S_1^{\text{SOGRAND}}(s)  $ for all possible $(q_1(s),q_2(s))$, $s \in \fb_2^{n-k}$, which by eq. \eqref{eq: all possible q_1, q_2 pairs} is only a subset of all possible pairs $(q_1,q_2) \in [2^n]^2$ with $q_1 < q_2$. Evaluating $S_1^{\text{EXACT}}(s)$ requires iterating through an exponentially large set, whereas evaluation of $S_1^{\text{SOGRAND}}(s)  $
requires only knowledge of $p_{N^n}$ and  $(q_1(s),q_2(s))$, which are byproducts of the decoded list. Nevertheless, Figs. \ref{fig:SO comparison CRC} and \ref{fig:SO comparison eBCH} illustrate that $S_1^{\text{SOGRAND}}$ approximates $S_1^{\text{EXACT}}(s)$ well.

In Figs. \ref{fig:SO comparison CRC} and \ref{fig:SO comparison eBCH}, and Lemma \ref{lem: structured code SO}, it can also be observed that the effect underlying the phenomenon discussed in Section \ref{sec: Shannon Theory}, that once the first likelihood rank is fixed, a late second rank increases confidence in the decoding, still occurs for binary linear codes.

\section{Conclusion}\label{sec: Conclusion}
In this paper we showed that the \ac{SOGRAND} soft-output formula \cite{yuanSOGRAND25}, which was originally described as an approximation, is exactly the conditional probability of a list entry being correct, given the likelihood ranks corresponding to the list. This holds for arbitrary ranking functions, including those that lead to approximations of \ac{ML} decoding. As such, it offers a new method for establishing results on the \ac{BLER}, cf. eq. \eqref{eq: BLER = 1- E[SO]}, particularly results conditioned on noise effects that identify codewords.

By studying the large-blocklength limit for \ac{ML} decoding, we identified a decision function $\Delta$, whose sign determines whether the conditional probability of the first list entry being correct converges to one or zero, and whose absolute value gives the exponential rate of convergence. This holds both above and below capacity.

The decision function depends only on the first, second and last likelihood ranks. For list size one, we recover the rank threshold $y^*$ of \cite{duffy_capacity-achieving_2019}, and show that for list size two, a second threshold $\bar{y}$ exists that separates the range of the first rank into three regions. Below $y^*$, the decoding concentrates on the correct codeword. In $(y^*, \bar{y})$, a late second rank recovers reliability, and above $\bar{y}$, the decoding  concentrates on a wrong codeword. We thus showed that a list decoder with at least two list entries contains strictly more reliability information than a single-codeword decoder. For the \ac{BSC}, we provided explicit expressions for $y^*$ and $\bar{y}$.

Using two distinct code structures as examples, we illustrated that for binary linear codes, the random codebook expression remains an accurate predictor of the exact conditional probability of correctness, and that the rescue phenomenon remains observable for structured codes.

 We performed a large-blocklength analysis in the hard-decision setting. However, the finite-blocklength results of Section \ref{sec: Soft Output exact} hold for an arbitrary ranking function and are therefore true in both the hard- and soft-decision settings. In a soft-decision setting, any generalization of the decoder we studied here, cf. Definition \ref{defi: decoder}, that implements \ac{ML} decoding \cite{solomon20} or approximations of \ac{ML} decoding \cite{duffy2022_ordered} requires a ranking function that depends on the soft information.
 An \ac{LDP} for guesswork with soft information is available \cite{LI_2019conditionalguesswork}, but its rate function averages over the realizations of the soft information and is thus not suitable for characterizing the reliability of a decoding conditioned on a fixed realization of the soft information. 
 
 One way forward is to consider the recently introduced Linear Reliability Channel \cite{Mariona2026}, which models soft information using discrete objects. There, an \ac{LDP} is established that holds also conditioned on a fixed realization of the soft information. Due to symmetries in the channel model, the resulting rate function is the same for every realization, which lends itself to the style of analysis presented here.

\bibliographystyle{IEEEtran}
\bibliography{references}

\end{document}